\documentclass[onecolumn,12pt,draft]{IEEEtran}
\usepackage{amsmath}
\usepackage[final]{graphicx}
\usepackage{amsmath,epsfig,cite}
\usepackage{graphicx}
\usepackage{subfigure}
\usepackage{amssymb}
\usepackage{mathrsfs}
\usepackage{color}
\newtheorem{lem}{Lemma}
\newtheorem{thry}{Theorem}
\newtheorem{conj}{Conjecture}
\newtheorem{corol}{Corollary}
\newtheorem{defn}{Definition}
\newtheorem{alg}{Algorithm}
\title{Analytical and Numerical Characterizations of Shannon Ordering for Discrete Memoryless Channels}

\author{Yuan Zhang and Cihan Tepedelenlio\u{g}lu
\thanks{The authors are with the School of Electrical,
Computer, and Energy Engineering, Arizona State University, Tempe,
AZ 85287, USA (Email: \{yzhang93, cihan\}@asu.edu). This work was
supported in part by the National Science Foundation under Grant CCF
1117041, and was presented in part at IEEE International Symposium on Information Theory, Jul. 2012.} }

\date{}

\begin{document}

\maketitle

\begin{abstract}
This paper studies several problems concerning channel inclusion,
which is a partial ordering between discrete memoryless channels
(DMCs) proposed by Shannon. Specifically, majorization-based
conditions are derived for channel inclusion between certain DMCs.
Furthermore, under general conditions, channel equivalence defined through Shannon ordering is shown to be the same as
permutation of input and output symbols. The determination of
channel inclusion is considered as a convex optimization problem,
and the sparsity of the weights related to the representation of the
worse DMC in terms of the better one is revealed when channel
inclusion holds between two DMCs. For the exploitation of this
sparsity, an effective iterative algorithm is established based on
modifying the orthogonal matching pursuit algorithm.
\end{abstract}

\section{Introduction}

The comparison between different communication channels has been a
long-standing problem since the establishment of Shannon theory.
Such comparisons are usually established through partial ordering
between two channels. Channel inclusion \cite{shannon58} is a
partial ordering defined for DMCs, when one DMC is obtained through
randomization at both the input and the output of another, and the
latter is said to include the former. Such an ordering between two
DMCs implies that for any code over the worse (included) DMC, there
exists a code of the same rate over the better (including) one with
a lower error rate. This enables ordering functions such as the error exponent or channel dispersion. Channel inclusion can be viewed as a generalization of the comparisons of statistical experiments
established in \cite{blackwell51,blackwell53}, in the sense that the
latter involves output randomization (degradation) but not input
randomization. There are also other kinds of channel ordering. For example,
more capable ordering and less noisy ordering \cite{korner75} enable the characterization of capacity regions of broadcast
channels. The partial ordering between finite-state Markov channels
is analyzed in \cite{eckford07,eckford09}. Our focus in this paper will be exclusively on channel inclusion as defined by Shannon \cite{shannon58}.

It is of interest to know how
it can be determined if one DMC includes another either
analytically, or numerically. To the best of our knowledge,
regarding the conditions for channel inclusion, the only results
beyond Shannon's paper \cite{shannon58} are provided in
\cite{helgert67,raginsky11}, and there is not yet any discussion on
the numerical characterization of channel inclusion in existing
literature. In this paper, we derive conditions for channel
inclusion between DMCs with certain special structure, as well as
channel equivalence, which complements the results in
\cite{helgert67} in useful ways, and relate channel inclusion to the
well-established majorization theory. In addition, we delineate the
computational aspects of channel inclusion, by formulating a convex
optimization problem for determining if one DMC includes another,
using a sparse representation. Compared to the conference version
\cite{ow12a}, this paper contains significant extensions.
As an example, for the purpose of obtaining a sparse solution, we develop an iterative algorithm based on
modifying orthogonal matching pursuit (OMP) and demonstrate its effectiveness. Moreover, we also find necessary and sufficient conditions for channel equivalence.

The rest of this paper is organized as follows. Section \ref{preli}
establishes the notation and describes existing literature. Section
\ref{cond_ch_inc} derives conditions for channel inclusion between
DMCs with special structure. Computational issues regarding channel inclusion are addressed in
Section \ref{comp}, followed by Section \ref{omp_mods} establishing
a sparsity-inducing algorithm for establishing channel inclusion. Section \ref{concl} concludes
the paper.

\section{Notations and Preliminaries} \label{preli}

Throughout this paper, a DMC is represented by a row-stochastic
matrix, i.e. a matrix with all entries being non-negative and each
row summing up to $1$. All the vectors involved are row vectors
unless otherwise specified. The entry of matrix $K$ with index $(i,j)$ and
the entry of vector $\mathbf{a}$ with index $i$ are denoted by
$[K]_{(i,j)}$ and $a_{(i)}$ respectively. The maximum (minimum)
entry of vector $\mathbf{a}$ is denoted by $\max \{ \mathbf{a} \}$
($\min \{ \mathbf{a} \}$), and $\mathbf{a} \geq {\bf 0}$ specifies
entry-wise non-negativity. The $i$-th row and $j$-th column of $K$
are denoted by $[K]_{(i,:)}$ and $[K]_{(:,j)}$ respectively. The set of indices from $n_1$ to $n_2 \geq n_1$ is denoted by $n_1:n_2$. The $n \times m$ matrix with all entries being $0$ (or $1$) is denoted by
$0_{n \times m}$ ($1_{n \times m}$). Also for convenience, we identify a DMC and its stochastic matrix, and apply
the terms ``square'', ``doubly stochastic'' and ``circulant'' for
matrices directly to DMCs. We next reiterate some of the definitions
and results in the literature related to this paper. We have the
following definitions.
\begin{defn} \label{defn1}
A DMC described by $n_1 \times m_1$ matrix $K_1$ is said to {\it include} \cite{shannon58}
another $n_2 \times m_2$ DMC $K_2$, denoted by $K_1 \supseteq K_2$ or $K_2 \subseteq
K_1$, if there exists a probability vector $\mathbf{g} \in
\mathbb{R}^{\beta}_+$ and $\beta$ pairs of stochastic matrices $\{
R_{\alpha},T_{\alpha} \}_{\alpha=1}^{\beta}$ such that
\begin{equation} \label{def_ci}
\sum_{\alpha=1}^{\beta} g_{(\alpha)} R_{\alpha} K_1 T_{\alpha}= K_2
.
\end{equation}
$K_1$ and $K_2$ are said to be {\it equivalent} if $K_1 \supseteq K_2$ and
$K_2 \supseteq K_1$. We say $K_2$ is strictly included in $K_1$,
denoted by $K_2 \subset K_1$, if $K_2 \subseteq K_1$ and $K_1
\nsubseteq K_2$. Intuitively, $K_2$ can be thought of as an input/output processed
version of $K_1$, with $g_{(\alpha)}$ being the probability that
$K_1$ is processed by $R_{\alpha}$ (input) and $T_{\alpha}$
(output). An operational interpretation of this definition is given
in Figure \ref{ch_inc4}, where, to ``simulate'' $K_2$, the channel
$R_{\alpha} K_1 T_{\alpha}$ is used with probability $g_{(\alpha)}$.
\end{defn}

\begin{figure}[!ht]
\begin{minipage}{1.0\textwidth}
\begin{center}
\includegraphics[height=30mm,keepaspectratio]{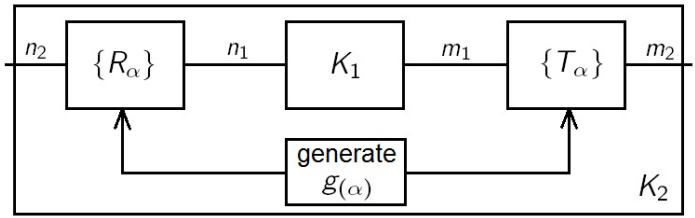}
\caption{Operational interpretation of $K_2 \subseteq K_1$, with
$K_1$ of size $n_1 \times m_1$ and $K_2$ of size $n_2 \times m_2$}
\label{ch_inc4}
\end{center}
\end{minipage}
\end{figure}

\begin{defn} \label{defn2}
A DMC $K_2$ is said to be a (output) degraded version
\cite{blackwell51,blackwell53} of another DMC $K_1$, if there exists
a stochastic matrix $T$ such that $K_1 T= K_2$.
\end{defn}

Note that output degradation in Definition \ref{defn2} is stronger than inclusion in Definition \ref{defn1}. There are several analytical conditions for channel inclusion derived in \cite{helgert67} for a special case of Definition
\ref{defn1} with $\beta=1$. Reference \cite{helgert67} considers two
kinds of DMCs, given by a $2 \times 2$ full-rank stochastic matrix
$P$, and an $n \times n$ stochastic matrix with identical diagonal
entries $p$ and identical off-diagonal entries $(1-p)/(n-1)$,
respectively. Necessary and sufficient conditions for $K_2= R K_1 T$ where $R$ and $T$ are stochastic matrices, are
derived for the cases in which $K_1$ and $K_2$ are of either of the
two kinds. Note that this assumes $\beta= 1$ in (\ref{def_ci}) and is with loss of generality. Conditions of inclusion for the general $\beta>1$ case have not yet been considered in the literature.

Channel inclusion can be equivalently defined with $R_{\alpha}$'s
and $T_{\alpha}$'s in Definition \ref{defn1} being stochastic
matrices in which all the entries are $0$ or $1$, as stated in
\cite{shannon58}, where $R_{\alpha}$'s and $T_{\alpha}$'s of this
kind are called {\it pure} matrices (or pure channels). This is easily corroborated
based on the fact that every stochastic matrix can be represented as
a convex combination of such pure matrices. This is due to the
fact that the set of stochastic matrices is convex and that $(0,1)$ stochastic matrices are extremal points of this set
\cite[Theorem 1]{schaefer82}. When $R_{\alpha}$ and $T_{\alpha}$ are pure matrices, the product $R_{\alpha} K_1 T_{\alpha}$ can be interpreted as a DMC whose input labels and output labels have been either permuted or combined. Therefore channel inclusion implies that the included DMC $K_2$ is in the convex hull of all such matrices, as seen in (\ref{def_ci}).

By considering $N$ uses of a DMC $K$, we equivalently have the DMC
$K^{\otimes N}$ which is the $N$-fold Kronecker product of $K$. We
have the following theorem, which was mentioned in \cite{shannon58} without a detailed proof.
\begin{thry} \label{th_kron}
$K_2 \subseteq K_1$ implies $K_2^{\otimes N} \subseteq K_1^{\otimes
N}$.
\end{thry}

\begin{proof}
See \ref{pr_kron}.
\end{proof}

As shown in \cite{shannon58}, $K_2 \subseteq K_1$ has the
implication that if there is a set of $M$ code words $\{ w_l
\}_{l=1}^M$ of length $N$, such that an error rate of ${\rm P_e}$ is
achieved with the code words being used with probabilities $\{ p_l
\}_{l=1}^M$ under $K_2$, then there exists a set of $M$ code words
of length $N$, such that an error rate of ${\rm P_e'} \leq {\rm
P_e}$ is achieved under $K_1$ with the code words being used with
probabilities $\{ p_l \}_{l=1}^M$. In \cite[p.116]{csiszarbook81},
this implication is stated as one DMC being {\it better in the
Shannon sense} than another (different from channel inclusion
ordering itself), and it is pointed out that $K_1 \supseteq K_2$ is
a sufficient but not necessary condition for $K_1$ to be better in
the Shannon sense than $K_2$, with the proof provided in
\cite{karmazin64}. This ordering of error rate in turn implies that
the capacity of $K_1$ is no less than the capacity of $K_2$, and the
same ordering holds for their error exponents.

Channel inclusion, as defined, is a {\it partial} order between two
DMCs: it is possible to have two DMCs $K_1$ and $K_2$ such that $K_1
\nsupseteq K_2$ and $K_2 \nsupseteq K_1$. For the purpose of making
it possible to compare an arbitrary pair of DMCs, a metric based on
the total variation distance, namely Shannon deficiency is
introduced in \cite{raginsky11}. In our notation, the Shannon
deficiency of $K_1$ with respect to $K_2$ is defined as
\begin{equation} \label{sh_defic}
\delta_S(K_1,K_2) \triangleq \inf_{\beta \in \mathbb{N}}
\inf_{\mathbf{g}; R_{\alpha},T_{\alpha} } \left\| \left(
\sum_{\alpha=1}^{\beta} g_{(\alpha)} R_{\alpha} K_1 T_{\alpha}- K_2
\right)^T \right\|_{\infty}
\end{equation}
where $\mathbf{g} \in \mathbb{R}^{\beta}_+$ is a probability vector,
$R_{\alpha}$'s and $T_{\alpha}$'s are stochastic matrices, $\| A
\|_{\infty} \triangleq \max_{i} \|[A]_{(i,:)}\|_1= \| A^T \|_1$ is
the $\infty$-norm of matrix $A$, and we impose matrix transpose
since we treat channel matrices as row-stochastic instead of
column-stochastic. Intuitively, the above Shannon deficiency
quantifies how far $K_1$ is from including $K_2$. Other useful deficiency-like quantities
are established in \cite{raginsky11} by substituting the total
variation distance with divergence-based metrics obeying a data
processing inequality between probability distributions.

\section{Analytical Conditions for Channel Equivalence and Inclusion} \label{cond_ch_inc}

In general, given two DMCs $K_1$ and $K_2$, there is no
straightforward method to check if one includes the other based on
their entries. Nevertheless, it is possible to characterize the
conditions for channel inclusion, for the cases in which both $K_1$ and $K_2$
have structure. In this section, we derive conditions for the
cases of doubly stochastic and circulant DMCs. For the case of
equivalence between two DMCs, we establish a necessary and
sufficient condition which is effectively applicable to any DMCs. We first define some useful notions.

\begin{defn} \label{defn3}
For two vectors $\mathbf{a}, \mathbf{b} \in \mathbb{R}^n$,
$\mathbf{a}$ is said to majorize (or dominate) $\mathbf{b}$, written
$\mathbf{a} \succ \mathbf{b}$, if and only if $\sum_{i=1}^k
a_{(i)}^{\downarrow} \geq \sum_{i=1}^k b_{(i)}^{\downarrow}$ for
$k=1,\dots,n-1$ and $\sum_{i=1}^n a_{(i)} = \sum_{i=1}^n b_{(i)}$,
where $a^{\downarrow}_{(i)}$ and $b^{\downarrow}_{(i)}$ are entries
of $\mathbf{a}$ and $\mathbf{b}$ sorted in decreasing order.
\end{defn}

\begin{defn} \label{defn4}
A circulant matrix is a square matrix in which the $i$-th row is
generated from cyclic shift of the first row by $i-1$ positions to
the right.
\end{defn}

\begin{defn} \label{defn5}
An $n \times n$ matrix $P$ is said to be doubly stochastic if the
following conditions are satisfied: {\bf (i)} $[P]_{(i,j)} \geq 0$
for $i,j=1,\dots,n$; {\bf (ii)} $\sum_i [P]_{(i,j)}=1$ for
$j=1,\dots,n$; {\bf (iii)} $\sum_j [P]_{(i,j)}=1$ for $i=1,\dots,n$.
\end{defn}

\begin{defn} \label{defn6}
A DMC is called symmetric if its rows are permutations of each
other, and its columns are permutations of each other \cite[p.190]{coverbook06}.
\end{defn}

It is easy to verify that if a symmetric DMC is square, then it must
be doubly stochastic. In the next section, we will focus mostly on
square DMCs (i.e. DMCs with equal size input and output
alphabets), and we assume this condition unless otherwise specified.

\subsection{Equivalence Condition between DMCs} \label{equiv1}

We address the general condition for two DMCs to be
equivalent, which has not been considered in the literature. By
imposing some mild assumptions, we have the following theorem which
gives the equivalence condition between two DMCs.
\begin{thry} \label{thr_equiv}
Let two DMCs $K_1$ and $K_2$ satisfying the following three assumptions
\begin{itemize}
\begin{item}
{\bf AS1} Capacity-achieving input distribution(s) contain no zero entry; That is, the capacity is not achieved if some of the input symbols is not used;
\end{item}
\begin{item}
{\bf AS2} There is no all-zero column, and no column being a multiple of another;
\end{item}
\begin{item}
{\bf AS3} If $K_1= P_1 K_1 P_2 D$ with permutation matrices $P_1$, $P_2$ and diagonal matrix $D$, then it is required that $P_1$, $P_2$ and $D$ are identity matrices; That is, by permuting the rows and columns of $K_1$, it is not possible to obtain a DMC whose columns are proportional to $K_1$. This property also applies for $K_2$.
\end{item}
\end{itemize}
Then a necessary and sufficient condition of $K_1$ being equivalent to $K_2$ is that $K_2= R K_1 T$ with $R$ and
$T$ being permutation matrices (thereby requiring $K_1$ and $K_2$ being of the same size $n \times m$).
\end{thry}

\begin{proof}
See \ref{pr_thr_equiv}.
\end{proof}

We have the following remarks about Theorem \ref{thr_equiv}. {\bf AS1} is verifiable through Blahut-Arimoto Algorithm \cite[ch. 13]{coverbook06}. Specifically, capacities can be obtained for the $n \times m$ DMC $K$ itself and the ones obtained by removing the $k$-th row from $K$ for $k=1,\dots,n$, and if the capacity is always reduced by removing a row, then the capacity-achieving input distribution of $K$ should have no zero entry. {\bf AS2} can be verified simply by inspection. Also, since DMCs are usually of small sizes in practice, it is viable to verify {\bf AS3} by inspection. For example, no column being a multiple of some entry-permuted version of another column makes a sufficient condition for {\bf AS3} to hold.

If two DMCs satisfying the above three assumptions are equivalent, there is an eigenvalue-based approach
for finding the permutation matrices without searching for all $n!m!$ such
permutations. Starting from $K_2= R K_1 T$, and $R^T= R^{-1}$, $T^T=
T^{-1}$ which is a property of permutation matrices, we have $K_2
K_2^T= R K_1 K_1^T R^{-1}$ which leads to the determination of $R$.
In order to do this, the first step is to perform the eigenvalue
decomposition: $K_1 K_1^T= Q_1 \Lambda Q_1^{-1}$ and $K_2 K_2^T= Q_2
\Lambda Q_2^{-1}$, where $\Lambda$ is a diagonal matrix, $Q_1$ and
$Q_2$ are both unitary matrices. Notice that it is necessary for
$K_1 K_1^T$ and $K_2 K_2^T$ to have the same set of eigenvalues,
otherwise $K_1$ and $K_2$ cannot be equivalent. Once we have these
decompositions, we can immediately obtain $R= Q_2 Q_1^{-1}$, which
is required to be a permutation matrix for $K_1$ and $K_2$ to be
equivalent. The determination of $T$ can also be made using the same
approach based on $K_2^T K_2= T^{-1} K_1^T K_1 T$, i.e. following
from the eigenvalue decompositions $K_1^T K_1= Q_3 \Sigma Q_3^{-1}$
and $K_2^T K_2= Q_4 \Sigma Q_4^{-1}$, $T= Q_4 Q_3^{-1}$ can be
obtained.

\subsection{Inclusion Conditions for Doubly Stochastic and Circulant DMCs}

Considering that doubly stochastic matrices have significant
theoretical importance, and doubly stochastic DMCs can be thought of
as a generalization of square symmetric DMCs, we first introduce the
following theorem
\begin{thry} \label{thr1}
Let $K_1$ and $K_2$ be $n \times n$ doubly stochastic DMCs, with
$\mathbf{w}_1$ and $\mathbf{w}_2$ being the $n^2 \times 1$ vectors
containing all the entries of $K_1$ and $K_2$ respectively. Then
$\mathbf{w}_2 \prec \mathbf{w}_1$ is a necessary condition for $K_2
\subseteq K_1$.
\end{thry}

\begin{proof}
See \ref{pr_thr1}.
\end{proof}

It should be pointed out that the above mentioned condition is not
sufficient. Otherwise, consider
\begin{equation} \label{ex55}
K_1= \begin{bmatrix} 1 & 2 & 3 & 4 & 5 \\ 5 & 1 & 2 & 3 & 4 \\ 4 & 5
& 1 & 2 & 3 \\ 3 & 4 & 5 & 1 & 2 \\ 2 & 3 & 4 & 5 & 1
\end{bmatrix} /15, K_2= \begin{bmatrix} 1 & 2 & 3 & 4 & 5 \\ 5 & 1 & 2 & 3 & 4 \\ 3 &
4 & 1 & 5 & 2 \\ 2 & 5 & 4 & 1 & 3 \\ 4 & 3 & 5 & 2 & 1
\end{bmatrix} /15
\end{equation}
it would be implied that $K_1$ and $K_2$ are equivalent. However,
based on Theorem \ref{thr_equiv}, it can be verified that $K_1$ and
$K_2$ are not equivalent since there do not exist permutation
matrices $R$ and $T$ such that $K_2= R K_1 T$ due to different sets of singular values of $K_1$ and $K_2$, thereby implying that
$\mathbf{w}_2 \prec \mathbf{w}_1$ is not sufficient for $K_2
\subseteq K_1$.

Consider the case of both $K_1$ and $K_2$ being $n \times n$
circulant, which are used to model channel noise captured by modulo arithmetic and has applications in discrete degraded interference
channels \cite{liu08}. We have the following result:
\begin{thry} \label{thr2}
Let $K_1$ and $K_2$ be $n \times n$ circulant DMCs, with vectors
$\mathbf{v}_1$ and $\mathbf{v}_2$ being their first rows,
respectively. Then for $K_2 \subseteq K_1$, a necessary condition is
$\mathbf{v}_2 \prec \mathbf{v}_1$. A sufficient condition is that
$\mathbf{v}_2$ can be represented as the circular convolution of
$\mathbf{v}_1$ and another probability vector $\mathbf{x}$ such that
$\mathbf{v}_1 \circledast \mathbf{x}= \mathbf{v}_2$, which is also
sufficient for output degradation.
\end{thry}

\begin{proof}
See \ref{pr_thr2}.
\end{proof}

It is clear that a $2 \times 2$ doubly stochastic DMC (also known as binary symmetric channel) is circulant
and characterized solely by the cross-over probability, thus the
condition for the inclusion between two $2 \times 2$ such DMCs boils
down to the comparison between their cross-over probabilities.
Furthermore, for $n=3,4$, it is easy to verify that if
an $n \times n$ symmetric DMC is not circulant, there is a circulant
DMC equivalent to it (for $n \geq 5$ there is no such guarantee as seen in (\ref{ex55})),
therefore we can conclude that
\begin{corol} \label{crl1}
For $n=3,4$, let $K_1$ and $K_2$ be $n \times n$ symmetric DMCs,
which are equivalent to circulant DMCs $K_1'$ and $K_2'$
respectively. Let $\mathbf{v}_1$ and $\mathbf{v}_2$ be the first
rows of $K_1'$ and $K_2'$ respectively. Then for $K_2 \subseteq
K_1$, a necessary condition is that $\mathbf{v}_2 \prec
\mathbf{v}_1$, while a sufficient condition is that $\mathbf{v}_2$
can be represented as the circular convolution of $\mathbf{v}_1$ and
another probability vector in $\mathbb{R}^n_+$.
\end{corol}

\begin{proof}
See \ref{pr_crl1}.
\end{proof}

We finally make a few remarks about inclusion between the binary symmetric channel (BSC) with cross over probability $p \leq 1/2$ and the binary erasure channel (BEC) with erasure probability $\epsilon$. It is well-known that BSC($p$) is a degraded version of BEC($\epsilon$) if and only if $0 \leq \epsilon \leq 2 p$ \cite[ch. 5.6]{elgamalbook12}. It can further be shown that BSC($p$) $\subseteq$ BEC($\epsilon$) if and only if $0 \leq \epsilon \leq 2 p$, while BEC($\epsilon$) $\subseteq$ BSC($p$) if and only if $p=0$. The ``if'' part follows directly from the fact that degradation implies channel inclusion. The ``only if'' part can be justified by the fact that inclusion is absent between BEC($\epsilon$) and BSC($p$) if $\epsilon > 2 p$ or $p> 0$.

\section{Computational Aspects of Channel Inclusion} \label{comp}

In Section \ref{cond_ch_inc}, we have established analytical
conditions for determining if a DMC with structure includes another.
It is also of interest to know how this can be determined
numerically when there is no structure. Furthermore, once it has
been determined that $K_2 \subseteq K_1$, it is desirable for
$g_{(\alpha)}$ probabilities in (\ref{def_ci}) to contain as many zeros as
possible to get a concise representation.

In this section, we provide a linear programming approach to calculating Shannon deficiency, which also enables checking if inclusion holds. For the cases in which channel inclusion is known to hold, we prove that sparse solutions exist and discuss
how this sparse solution for $\mathbf{g}$ can be obtained through sparse recovery techniques, such as orthogonal matching pursuit.

We first take a look at determining if $K_2
\subseteq K_1$ through convex optimization. For $K_1$ of size $n_1
\times m_1$ and $K_2$ of size $n_2 \times m_2$, the problem can be
formulated as
\begin{equation} \label{opt_prob1}
\begin{split}
\textrm{minimize } & \left\| \sum_{\alpha=1}^{\beta} g_{(\alpha)}
R_{\alpha} K_1 T_{\alpha}- K_2 \right\|_1 \\
\textrm{subject to } & \sum_{\alpha=1}^{\beta} g_{(\alpha)}=1,
g_{(\alpha)} \geq 0
\end{split}
\end{equation}
with variables $\mathbf{g} \in \mathbb{R}^{\beta}$, where
$R_{\alpha}$ is $n_2 \times n_1$, and $T_{\alpha}$ is
$m_1 \times m_2$ stochastic matrices for $\alpha= 1,\dots,\beta$, and $K_2 \subseteq K_1$ is
determined if the optimal value is zero. As mentioned in Section
\ref{preli}, $R_{\alpha}$'s and $T_{\alpha}$'s can be equivalently
treated as pure channels, so there are at most $n_1^{n_2} m_2^{m_1}$
different $\{ R_{\alpha},T_{\alpha} \}$ pairs, and consequently
there are finitely many $g_{(\alpha)}$'s involved in the problem
(\ref{opt_prob1}). It is easy to see that
(\ref{opt_prob1}) is a convex optimization problem, and it can be
re-formulated as a linear programming problem with variables
$g_{(\alpha)}$ and an $n_2 \times 1$ vector ${\bf c}$
\begin{equation} \label{opt_prob2}
\begin{split}
& \textrm{minimize } 1^T {\bf c} \\
& \textrm{subject to } -{\bf c} \leq \left[
\sum_{\alpha=1}^{\beta} g_{(\alpha)} R_{\alpha} K_1 T_{\alpha}- K_2 \right]_{(:,j)} \leq {\bf c}, \textrm{ for } j=1,\dots,m_2, \\
& \sum_{\alpha=1}^{\beta} g_{(\alpha)}=1, g_{(\alpha)} \geq 0 .
\end{split}
\end{equation}
We also notice that the optimal value of (\ref{opt_prob2}) provides
a way to evaluate the Shannon deficiency of $K_1$ with respect to
$K_2$.

In the above analysis, the maximum number of $\{ R_{\alpha},T_{\alpha} \}$ pairs, given by $n_1^{n_2} m_2^{m_1}$ (or
$(n!)^2$ if both $K_1$ and $K_2$ are $n \times n$ doubly
stochastic), grows very rapidly with the sizes of $K_1$ and $K_2$.
With $K_2 \subseteq K_1$ already determined, it is natural to ask if
(\ref{def_ci}) can hold with some reduced number of $\{
R_{\alpha},T_{\alpha} \}$ pairs. In other words, we seek to have a
sparse solution of $\mathbf{g}$. We have the following theorem
regarding the sparsity of $\mathbf{g}$ given $K_2 \subseteq K_1$, based on Carath\'{e}odory's theorem \cite[p.155]{rockafellarbook70}.
\begin{thry} \label{thr4}
For two DMCs $K_1$ of size $n_1 \times m_1$ and $K_2$ of size $n_2
\times m_2$, if $K_2 \subseteq K_1$, there exist a probability
vector $\mathbf{g} \in \mathbb{R}^{\beta}_+$ and $\beta$ pairs
of stochastic matrices $\{ R_{\alpha},T_{\alpha}
\}_{\alpha=1}^{\beta}$ such that (\ref{def_ci}) holds with $\beta \leq n_2 (m_2-1)+1$. If both $K_1$ and $K_2$ are $n \times
n$ doubly stochastic, the number of necessary $\{
R_{\alpha},T_{\alpha} \}$ pairs in (\ref{def_ci}) can be improved as
$\beta \leq (n-1)^2+1$.
\end{thry}

\begin{proof}
See \ref{pr_thr4}.
\end{proof}

It is well-known that a typical approach to recover a sparse signal
vector from its linear measurements is compressed sensing with
$\ell_1$ norm minimization (also known as basis pursuit). To
apply this approach to our problem, we can formulate it as
\begin{equation} \label{opt_prob3}
\begin{split}
\textrm{minimize } & \sum_{\alpha=1}^{\beta} |g_{(\alpha)}| \\
\textrm{subject to } & \sum_{\alpha=1}^{\beta} g_{(\alpha)}
R_{\alpha} K_1 T_{\alpha}= K_2
\end{split}
\end{equation}
with variables $\mathbf{g} \in \mathbb{R}^{\beta}$. It is easy to
prove that the optimal $\mathbf{g}$ always comes out non-negative
given $K_2 \subseteq K_1$. However, (\ref{opt_prob3}) does not
necessarily give a sparse solution for $\mathbf{g}$. As pointed out in \cite{cohen11} which addresses the solvability of a sparse probability vector based on linear measurements through $\ell_1$ norm minimization, in order for the sparse probability vector to be solvable, the number of independent measurements needs to be at least two times the sparsity level. In our case this is not satisfied, since the number of independent equations ($n_2 (m_2-1)$ or $(n-1)^2$) in the
constraints in (\ref{opt_prob3}) is usually less than $2 \beta$
(which can be up to $2n_2 (m_2-1)+2$ or $2(n-1)^2+2$). There are also other sparsity-inducing numerical methods such as
matching pursuit, which will be addressed in the next section.

\section{Channel Inclusion through OMP} \label{omp_mods}

Orthogonal matching pursuit (OMP) \cite{tropp07} and its variants
are widely investigated in the literature for sparse solutions of
linear equations. OMP algorithm gives a possibly sub-optimal
solution to the following problem with vector ${\bf g}$ being the
variable
\begin{equation} \label{opt_prob4a}
\begin{split}
& \textrm{minimize } Q({\bf g})= \| {\bf h}- A {\bf g} \|_2^2\\
& \textrm{subject to } \| {\bf g} \|_0 \leq s
\end{split}
\end{equation}
through which the known upper bound $s$ of sparsity level is
exploited. Notice that the standard OMP algorithm does not impose
the constraint ${\bf g} \geq {\bf 0}$. In the context of the channel
inclusion problem, $A$ is a $n_2 m_2 \times n_1^{n_2} m_2^{m_1}$
matrix with its $\alpha$-th column $[A]_{(:,\alpha)}= \mathrm{vec}
(R_{\alpha} K_1 T_{\alpha})$ (i.e. $[A]_{(:,\alpha)}$ is the
vectorized version of $R_{\alpha} K_1 T_{\alpha}$ by stacking its
columns in a vector), and ${\bf h}= \mathrm{vec} (K_2)$. Moreover,
we have the additional constraint ${\bf g} \geq {\bf 0}$ so that
(\ref{opt_prob4a}) becomes
\begin{equation} \label{opt_prob4}
\begin{split}
& \textrm{minimize } Q({\bf g})= \| {\bf h}- A {\bf g} \|_2^2\\
& \textrm{subject to } \| {\bf g} \|_0 \leq s, \hspace{2mm} {\bf g} \geq {\bf 0}
\end{split}
\end{equation}
where $s= n_2 (m_2-1)+ 1$. Note that if inclusion is present the
solution will automatically satisfy $\| {\bf g} \|_1= 1$, without adding this as an extra constraint. The problem in
(\ref{opt_prob4}) is related to (\ref{opt_prob1}) and
(\ref{opt_prob2}) in the sense that if the optimal value of
(\ref{opt_prob4}) is zero, the solution of (\ref{opt_prob4}) is also
the solution of (\ref{opt_prob1}) and (\ref{opt_prob2}).

To introduce briefly, OMP algorithm finds a sparse solution of
(\ref{opt_prob4a}) by selecting columns of $A$ having inner products
with the residue ${\bf h}- A {\bf g}$ with a large magnitude. This
requires taking the absolute value of the inner products in solving
(\ref{opt_prob4a}), followed by solving a least-square (LS) problem. However, to solve (\ref{opt_prob4}) we require
${\bf g}$ to have non-negative entries. We will modify the standard
OMP to encourage this result by not taking the absolute value of the
inner products, which is shown in Theorem \ref{thr2_alg2} to be a necessary condition for the LS solution to be non-negative in each entry.

In this section, assuming channel inclusion is present, we introduce OMP-like algorithms which solve for a sparse probability vector involved in channel inclusion. The established algorithm is also applicable to
other problems (e.g. solving for sparse probability vector based on moments of the discrete random variable \cite{cohen11}) with the objective of solving for non-negative vectors, and we will describe it in general terms. Unlike the
standard OMP algorithm which operates without positivity constraints
on the solution, the algorithms established here aim to find a
non-negative sparse solution of ${\bf g}$ based on ${\bf h}$ and $A$. For this purpose, modifications are needed in
our algorithms compared to the standard OMP algorithm which solves
(\ref{opt_prob4a}), in order to solve the problem in (\ref{opt_prob4}). For
example, standard OMP relies on choosing the inner product with the
largest absolute value, while our algorithms consider the signed inner
product; standard OMP makes one attempt per iteration for the
least-square solution, while it is possible for our algorithm to
make multiple attempts. This is because we insist that at each iteration the LS solution yields non-negative entries, which depends on the column chosen at the current iteration. If the LS solution provides some negative entries, instead of projecting the solution to the non-negative orthant, we start over and select a new column with a positive inner product. This preserves the orthogonality of the residue with all the selected columns. The details of our algorithm is given as follows.

\noindent\rule{\textwidth}{1pt}

\begin{alg} \label{alg1}
The modified OMP algorithm for retrieving non-negative sparse vector
${\bf g}$ from $A{\bf g}= {\bf h}$
with known upper bound of sparsity level $s$ consists of the
following

{\bf Inputs}:
\begin{itemize}
\begin{item}
An $p \times q$ matrix $A$ with $p \ll q$
\end{item}
\begin{item}
An $p \times 1$ vector ${\bf h}$ which consists of noise-free linear
measurements of ${\bf g}$
\end{item}
\begin{item}
The known upper bound of sparsity level $s$ of the non-negative
vector ${\bf g}$ (in general it is $p$; for the channel inclusion
problem, it is as specified in Theorem \ref{thr4})
\end{item}
\begin{item}
Tolerance $\epsilon$, for error being essentially zero
\end{item}
\end{itemize}

{\bf Outputs}:
\begin{itemize}
\begin{item}
A flag $f$ for a solution being found ($f=1$) or not found ($f=0$)
\end{item}
\begin{item}
The number $s_1$ of iterations for the residue to
become essentially zero (if $f=1$)
\end{item}
\begin{item}
A set (vector) $\Lambda_{s_1}$ of column indices for $A$, $\Lambda_{s_1}
\subseteq \{ 1, 2, \dots, q \}$ (if $f=1$)
\end{item}
\begin{item}
An $s_1 \times 1$ vector ${\bf g}_{s_1}$ (if $f=1$)
\end{item}
\end{itemize}

{\bf Procedure}:

Initialize the residue to ${\bf r}_0= {\bf h}$, the set of indices to
$\Lambda= 0_{1 \times s}$, the matrix containing the columns of $A$
which are selected to $A_{\rm sel}= 0_{p \times s}$, the inner product
vector to $P= 0_{1 \times q}$, and the iteration counter $t=1$. The
remaining steps are given in pseudo-code as follows:

01: while $t \leq s$ and $\| {\bf r}_{t-1} \|_{\infty} \geq
\epsilon$

02: \hspace{5mm} $P= {\bf r}_{t-1}^T A$; \hfill{} $\triangleright$
inner product generation

03: \hspace{5mm} ${\bf g}_t= -1_{t \times 1}$; \hfill{}
$\triangleright$ initializing the sparse vector

04: \hspace{5mm} while $\min \{ {\bf g}_t \}< 0$ and $\max \{ P \}>
0$

05: \hspace{10mm} $\lambda_t= \arg\max_j P_{(j)}$; $\Lambda_{(t)}=
\lambda_t$; \hfill{} $\triangleright$ locating the largest inner
product

06: \hspace{10mm} $[A_{\rm sel}]_{(:,t)}= [A]_{(:,\lambda_t)}$;
\hfill{} $\triangleright$ selecting a new column of $A$
corresponding to $\lambda_t$

07: \hspace{10mm} ${\bf g}_t= \arg\min_{\bf g} \| {\bf h}- [A_{\rm
sel}]_{(:,1:t)} {\bf g} \|_2^2$; \hfill{} $\triangleright$ solving a
least-square problem

08: \hspace{10mm} $P_{(\lambda_t)}= -1$; \hfill{} $\triangleright$
marking index $\lambda_t$ as attempted to avoid multiple attempts

09: \hspace{5mm} end;

10: \hspace{5mm} ${\bf r}_{t}= {\bf h}- [A_{\rm sel}]_{(:,1:t)} {\bf
g}_t$; $t= t+1$; \hfill{} $\triangleright$ updating residue for the
next iteration

11: end;

Finally, set $f=1$ if $\min \{ {\bf g}_{t-1} \} \geq 0$ and $\| {\bf
r}_{t-1} \|_{\infty}< \epsilon$, otherwise $f=0$. With $f=1$, the
other outputs are $s_1= t-1$, $\Lambda_{s_1}= \Lambda_{(1:s_1)}$,
and ${\bf g}_{s_1}$ is as given at the termination of the
iterations. The $j$-th entry of ${\bf g}_{s_1}$ is the $\lambda_j$-th entry of ${\bf g}$ and all other entries of ${\bf g}$ are zero.

\end{alg} \noindent\rule{\textwidth}{1pt}

Notice that the ``while'' loop starting Line $4$ in Algorithm
\ref{alg1} always terminates because there are always finitely many
positive inner products available for selection. Algorithm
\ref{alg1} inherits the keys steps directly from the standard OMP
algorithm, as seen from Lines $2$, $6$, $7$ and $10$. It differs
from the standard OMP algorithm in that it aims to find a
non-negative least-square solution at each iteration unless all the
positive inner products are depleted, which is reflected by Line
$4$. As seen from numerical simulations, it has a very low rate of
failure in the sense that it returns several $f=0$ out of a very
large number of tests in which channel inclusion is present. An
illustration of this is given in Figure \ref{alg1_fail}, which shows
the rate of failure of Algorithm \ref{alg1}, with $\beta= 1,2,3,4,5$
and randomly generated stochastic matrices $K_1$, $\{ R_{\alpha},
T_{\alpha} \}_{\alpha=1}^{\beta}$ as well as probability vector $\{
g_{\alpha} \}_{\alpha=1}^{\beta}$. Specifically, all matrix and vector entries are generated according to uniform distribution in $[0,1]$ and then normalized to satisfy probability constraint. We can also observe that the
rates of failure are very close for different values of $\beta$.

\begin{figure}[!ht]
\begin{minipage}{1.0\textwidth}
\begin{center}
\includegraphics[height=9cm,keepaspectratio]{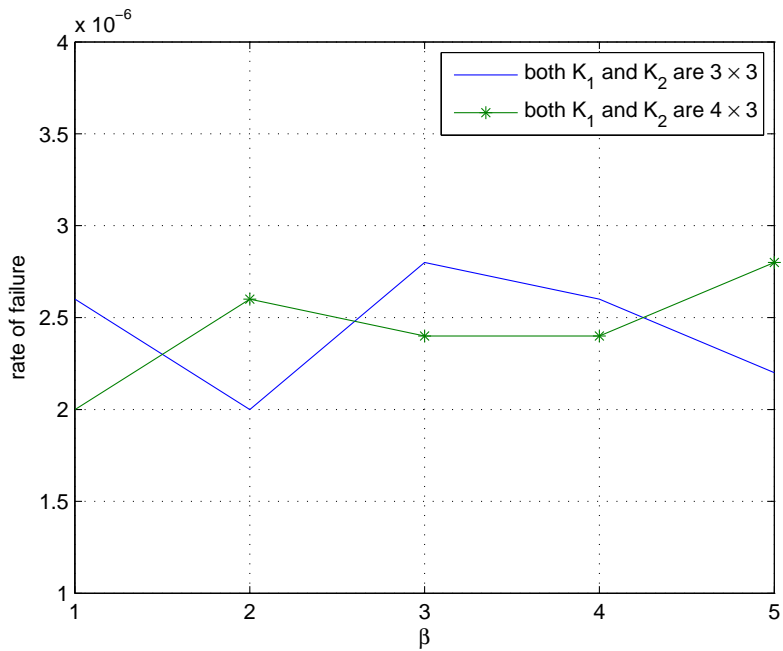}
\caption{Rate of failure of Algorithm \ref{alg1} with randomly
generated stochastic matrices $K_1$ and $K_2 \subseteq K_1$}
\label{alg1_fail}
\end{center}
\end{minipage}
\end{figure}

Failures occur if the algorithm produces a vector ${\bf g}_{s_1}$ that has negative entries. It is natural to ask why Algorithm \ref{alg1} produces failures. We rule out the selection of a positive inner product (as reflected in
Lines $4$ and $5$) from being the reason, as justified by the
following theorem.
\begin{thry} \label{thr2_alg2}
In Algorithm \ref{alg1}, the selection of a positive inner product
(as reflected in Lines $4$ and $5$) is necessary for the
least-square solution (in Line $7$) to be non-negative. Moreover, at each iteration, vector $P$ always has at least one positive entry, so that a (not yet selected) column of $A$ having a positive inner product with the residue is always possible.
\end{thry}

\begin{proof}
See \ref{pr_thr2_alg2}.
\end{proof}

Theorem \ref{thr2_alg2} implies that no mistake is made by not
considering the negative inner products. Thus we believe that the
failures produced by Algorithm \ref{alg1} are due to the fact that
not all the possible selections of inner products are attempted.
Going one step further from Algorithm \ref{alg1}, it is desirable to establish an improved algorithm which is always
successful. We now describe the algorithm which can be proved based on a forthcoming
conjecture to be always successful in solving for sparse
probability vector involved in channel inclusion, provided that
inclusion is present.

\noindent\rule{\textwidth}{1pt}

\begin{alg} \label{alg2}
The modified OMP algorithm for retrieving non-negative sparse vector
${\bf g}$ from $A{\bf g}= {\bf h}$ with known upper bound of sparsity level $s$ consists of the following

{\bf Inputs}:
\begin{itemize}
\begin{item}
An $p \times q$ matrix $A$ with $p \ll q$
\end{item}
\begin{item}
An $p \times 1$ vector ${\bf h}$ which consists of noise-free linear
measurements of ${\bf g}$
\end{item}
\begin{item}
The known upper bound of sparsity level $s$ of the non-negative
vector ${\bf g}$ (in general it is $p$; for the channel inclusion
problem, it is as specified in Theorem \ref{thr4})
\end{item}
\begin{item}
Tolerance $\epsilon$, for error being essentially zero
\end{item}
\end{itemize}

{\bf Outputs}:
\begin{itemize}
\begin{item}
A flag $f$ for a solution being found ($f=1$) or not found ($f=0$)
\end{item}
\begin{item}
The number $s_1$ of iterations for the residue to
become essentially zero (if $f=1$)
\end{item}
\begin{item}
A set (vector) $\Lambda_{s_1}$ of column indices for $A$, $\Lambda_{s_1}
\subseteq \{ 1, 2, \dots, q \}$ (if $f=1$)
\end{item}
\begin{item}
An $s_1 \times 1$ vector ${\bf g}_{s_1}$ (if $f=1$)
\end{item}
\end{itemize}

{\bf Procedure}:

Initialize the residue to ${\bf r}_0= {\bf h}$, the set of indices to
$\Lambda= 0_{1 \times s}$, the matrix containing the columns of $A$
which are selected to $A_{\rm sel}= 0_{p \times s}$, the inner product
matrix to $P= 0_{s \times q}$, and the iteration counter $t=1$. For observation purpose we also count the actual number of iterations $t_{\rm act}$, which is initialized as zero. The remaining steps are given in pseudo-code as follows:

01: while $1 \leq t \leq s$ and $\| {\bf r}_{t-1} \|_{\infty} \geq
\epsilon$

02: \hspace{5mm} if $\max \{ [P]_{(t,:)} \} \leq 0$ and $\min \{
[P]_{(t,:)} \} < 0$

03: \hspace{10mm} $[P]_{(t,:)}= 0_{1 \times q}$; $t=t-1$; \hfill{}
$\triangleright$ resetting inner product and tracing back

04: \hspace{5mm} else

05: \hspace{10mm} if $[P]_{(t,:)}== 0_{1 \times q}$

06: \hspace{15mm} $[P]_{(t,:)}= {\bf r}_{t-1}^T A$; \hfill{}
$\triangleright$ inner product generation

07: \hspace{10mm} end;

08: \hspace{10mm} ${\bf g}_t= -1_{t \times 1}$; \hfill{}
$\triangleright$ initializing the sparse vector

09: \hspace{10mm} while $\min \{ {\bf g}_t \}< 0$ and $\max \{
[P]_{(t,:)} \}> 0$

10: \hspace{15mm} $\lambda_t= \arg\max_j [P]_{(t,j)}$; $\Lambda_{(t)}=
\lambda_t$; \hfill{} $\triangleright$ locating the largest inner
product

11: \hspace{15mm} $[A_{\rm sel}]_{(:,t)}= [A]_{(:,\lambda_t)}$;
\hfill{} $\triangleright$ selecting a new column of $A$
corresponding to $\lambda_t$

12: \hspace{15mm} ${\bf g}_t= \arg\min_{\bf g} \| {\bf h}- [A_{\rm
sel}]_{(:,1:t)} {\bf g} \|_2^2$; \hfill{} $\triangleright$ solving a
least-square problem

13: \hspace{15mm} $[P]_{(t,\lambda_t)}= -1$; \hfill{}
$\triangleright$ marking index $\lambda_t$ as attempted to avoid
multiple attempts

14: \hspace{10mm} end;

15: \hspace{10mm} if $\min \{ {\bf g}_t \}< 0$

16: \hspace{15mm} $[P]_{(t,:)}= 0_{1 \times q}$; $t=t-1$; \hfill{}
$\triangleright$ resetting inner product and tracing back

17: \hspace{10mm} else

18: \hspace{15mm} ${\bf r}_{t}= {\bf h}- [A_{\rm sel}]_{(:,1:t)}
{\bf g}_t$; $t= t+1$; \hfill{} $\triangleright$ updating residue for
the next iteration

19: \hspace{10mm} end;

20: \hspace{5mm} end;

21: \hspace{5mm} $t_{\rm act}= t_{\rm act}+ 1$;

22: end;

Finally, set $f=1$ if $t \geq 1$ and $\| {\bf r}_{t-1} \|_{\infty}<
\epsilon$, otherwise $f=0$. With $f=1$, the other outputs are $s_1=
t-1$, $\Lambda_{s_1}= \Lambda_{(1:s_1)}$, and ${\bf g}_{s_1}$ is as
given at the termination of the iterations. The $j$-th entry of ${\bf g}_{s_1}$ is the $\lambda_j$-th entry of ${\bf g}$ and all other entries of ${\bf g}$ are zero.

\end{alg} \noindent\rule{\textwidth}{1pt}

Algorithm \ref{alg2} differs from Algorithm \ref{alg1} primarily in
the following two aspects: the inner product is changed from a
vector into a matrix, as reflected in Line $6$, for the purpose of
recalling the values of inner products involved in the past
iterations. Moreover, the iteration may go backward, as reflected by Lines $3$ and
$16$, in the sense that the most recently added columns of $A_{\rm sel}$ may be deleted in order to ``backtrack''. In Algorithm \ref{alg2}, the iteration proceeds at $t$ when a new column of $A$ can be found, such that with $[A_{\rm
sel}]_{(:,t)}$ updated as this new column, it follows that ${\bf
g}_t= \arg\min_{\bf g} \| {\bf h}- [A_{\rm sel}]_{(:,1:t)} {\bf g}
\|_2^2 \geq {\bf 0}$, i.e. the least-square solution of the sparse vector
is non-negative in each entry; otherwise, the iteration traces back and updates
the selection of $[A_{\rm sel}]_{(:,t-1)}$, for the purpose of
making it possible to find $[A_{\rm sel}]_{(:,t)}$ such that ${\bf
g}_t= \arg\min_{\bf g} \| {\bf h}- [A_{\rm sel}]_{(:,1:t)} {\bf g}
\|_2^2 \geq {\bf 0}$. When the iteration proceeds, the residue is updated
for inner product generation in the next iteration; when the
iteration traces back, the inner product is reset, in order to
enable its re-generation when the iteration proceeds to this step a
second time.

We now introduce the following conjecture which will lead to the
effectiveness (to be proved in Theorem \ref{thr1_alg2}) of Algorithm \ref{alg2}.
\begin{conj} \label{prop1}
Let $G$ be a matrix with all entries being non-negative and all
columns being linearly independent. There exists at least one column
${\bf g}_*$ of $G$ such that, with $G_*$ obtained by excluding ${\bf
g}_*$ from $G$, $\hat{\bf x}:= \arg\min_{\bf x} \| {\bf g}_*- G_*
{\bf x} \|_2^2$ has non-negative entries.
\end{conj}

Conjecture \ref{prop1} points out that among several linearly
independent non-negative vectors, there is at least one of them,
whose orthogonal projection onto the hyperplane defined by the other
vectors is a conic combination of those vectors. In the following, we show the effectiveness of Algorithm \ref{alg2},
as stated in Theorem \ref{thr1_alg2}.

\begin{thry} \label{thr1_alg2}
If Conjecture \ref{prop1} holds, then Algorithm \ref{alg2} does not fail, i.e. $f=1$ is returned when inclusion is present.
\end{thry}

\begin{proof}
See \ref{pr_thr1_alg2}.
\end{proof}

For Algorithm \ref{alg2} to fail, $A_{\rm sel}$ must have no column of $A$, and all the columns of $A$ have been attempted but none of them is selected eventually. These possible multiple attempts all occur at $t=1$, when $A_{\rm sel}$ has no column of $A$. Theorem \ref{thr1_alg2} effectively rules out this possibility, and implies that Algorithm \ref{alg2} is
guaranteed to work by searching for a non-negative least-square
solution at each iteration, in the sense that there exists a path of
iterations, in which an atom (a column of $A$) associated with a positive inner
product is selected at each iteration, eventually leading to a
solution with all entries of ${\bf g}_{s_1}$ being non-negative. Essentially, Theorem \ref{thr1_alg2} implies that by only focusing on the selection of a new column which results in a non-negative intermediate solution ${\bf g}_t$ (as reflected in Lines $9$ and $15$ of Algorithm \ref{alg2}), we do not have the risk of driving Algorithm \ref{alg2} into failure. If Algorithm \ref{alg1} or \ref{alg2} terminates with $f=1$, the residue can be treated as zero. From this, it can be shown that ${\bf g}_{s_1}$ is a probability vector: consider the product $1_{1 \times n_2 m_2} ([A_{\rm sel}]_{(:,1:s_1)} {\bf g}_{s_1}-{\bf h})= 0$, we have $n_2 1_{1 \times s_1} {\bf g}_{s_1}= n_2$, which shows that the entries of ${\bf g}_{s_1}$ sum up to $1$, i.e. a sparse {\it probability vector} relating $K_1$ and $K_2$ is obtained.

By performing the same numerical tests (i.e. for both $K_1$ and $K_2$ being $3 \times 3$ or $4 \times 3$, with randomly generated stochastic matrices $K_1$, $\{ R_{\alpha}, T_{\alpha} \}_{\alpha=1}^{\beta}$ as well as probability vector $\{ g_{\alpha} \}_{\alpha=1}^{\beta}$) as performed on Algorithm \ref{alg1}, it is observed that Algorithm \ref{alg2} produces no failure in $5 \times 10^6$ tests for each case. It is also seen that Algorithm \ref{alg2} does not invoke many backtracks in practice if inclusion is present, which is as expected given the fact that Algorithm \ref{alg1} has a very low rate of failure.

Furthermore, starting from two given DMCs $K_1$ and $K_2$ without knowing the presence or absence of inclusion, for the purpose of determining if inclusion is present, the $\ell_1$ minimization approach given by (\ref{opt_prob2}) should be used since it provides guaranteed correctness about the presence or absence of inclusion. Once the presence of inclusion is identified, for the purpose of obtaining a sparse probability vector relating $K_1$ and $K_2$, Algorithm \ref{alg1} can be used first, and if Algorithm \ref{alg1} does not return a sparse probability vector as desired, Algorithm \ref{alg2} becomes the choice for this purpose. Although we do not have a proof that Algorithm \ref{alg2} does not incur a lot of backtracking, we known empirically that it is the case, and thus Algorithm \ref{alg2} is favorable in the sense that it makes a more effective and less complex approach for obtaining a sparse solution than $\ell_1$ minimization approach.

\section{Conclusions} \label{concl}

In this paper, we investigate the characterization of channel
inclusion between DMCs through analytical and numerical approaches.
We have established several conditions for equivalence between DMCs,
and for inclusion between DMCs with structure including doubly
stochastic, circulant, and symmetric DMCs. We formulate a linear
programming problem leading to the quantitative result on how far is
one DMC apart from including another, which has an implication on
the comparison of their error rate performance. In addition, for the
case in which one DMC includes another, by using Carath\'{e}odory's theorem, we derive an upper bound
for the necessary number of pairs of pure channels involved in the
representation of the worse DMC in terms of the better one, which is
significantly less than the maximum possible number of such pairs.
This kind of sparsity implies reduced complexity of finding the
optimal code for the better DMC based on the code for the worse one.
By modifying the standard OMP algorithm, an iterative algorithm that
exploits this sparsity is established, which is seen to be
significantly less complex than basis pursuit and produces no failure
in determining the presence or absence of channel inclusion. Such effectiveness in determining the presence or absence of channel inclusion is proved with the help of a conjecture.

\renewcommand\thesection{Appendix \Alph{section}}
\setcounter{section}{0}
%\appendices

\section{Proof of Theorem \ref{th_kron}} \label{pr_kron}

Given (\ref{def_ci}), it follows that
\begin{equation} \label{ci_kron}
\left( \sum_{\alpha=1}^{\beta} g_{(\alpha)} R_{\alpha} K_1
T_{\alpha} \right)^{\otimes N} = K_2^{\otimes N}
\end{equation}
Based on the bilinearity of Kronecker product, the left hand side of
(\ref{ci_kron}) can be expanded into the summation of $\beta^N$
terms, which are all in the form of $\Sigma_{j(i): \{ 1,\dots, N \}
\rightarrow \{ 1,\dots, \beta \}} (\prod_{i=1}^N g_{j(i)})
[(R_{j(1)} K_1 T_{j(1)}) \otimes \cdots \otimes (R_{j(N)} K_1
T_{j(N)})]$ where the summation is over all possible functions
$j(i)$: $\{ 1,\dots, N \}$ $\rightarrow$ $\{ 1,\dots, \beta \}$.
Based on the mixed-product property of Kronecker product, we have
\begin{equation} \label{mprod_kron_ci}
\begin{split}
(R_{j(1)} K_1 T_{j(1)}) \otimes (R_{j(2)} K_1 T_{j(2)}) &= (R_{j(1)}
K_1) \otimes (R_{j(2)} K_1) (T_{j(1)} \otimes T_{j(2)})\\
&= (R_{j(1)} \otimes R_{j(2)}) K_1^{\otimes 2} (T_{j(1)} \otimes
T_{j(2)})
\end{split}
\end{equation}
By applying (\ref{mprod_kron_ci}) repeatedly, it follows that
$(\prod_{i=1}^N g_{j(i)}) [(R_{j(1)} K_1 T_{j(1)}) \otimes \cdots
\otimes (R_{j(N)} K_1 T_{j(N)})]= (\prod_{i=1}^N g_{j(i)}) (R_{j(1)}
\otimes \cdots \otimes R_{j(N)}) K_1^{\otimes N} (T_{j(1)} \otimes
\cdots \otimes T_{j(N)})$, which in turn implies that the left hand
side of (\ref{ci_kron}) expands into $\beta^N$ terms in the form of
$(\prod_{i=1}^N g_{j(i)}) (R_{j(1)} \otimes \cdots \otimes R_{j(N)})
K_1^{\otimes N} (T_{j(1)} \otimes \cdots \otimes T_{j(N)})$, and
thus $K_2^{\otimes N} \subseteq K_1^{\otimes N}$ by Definition
\ref{defn1}.

\section{Proof of Theorem \ref{thr_equiv}} \label{pr_thr_equiv}

Consider $K_1$ of size $n_1 \times m_1$ and $K_2$ of size $n_2
\times m_2$. By Definition \ref{defn1} we have
\begin{subequations}
\begin{equation} \label{k1_equiv}
K_1= \sum_{\alpha_1=1}^{\beta_1} g_{1(\alpha_1)} R_{1,\alpha_1} K_2
T_{1,\alpha_1}
\end{equation}
\begin{equation} \label{k2_equiv}
K_2= \sum_{\alpha_2=1}^{\beta_2} g_{2(\alpha_2)} R_{2,\alpha_2} K_1
T_{2,\alpha_2}
\end{equation}
\end{subequations}
with $R_{1,\alpha_1}$'s of size $n_1 \times n_2$, $T_{1,\alpha_1}$'s
of size $m_2 \times m_1$,  $R_{2,\alpha_2}$'s of size $n_2 \times
n_1$,  $T_{2,\alpha_2}$'s of size $m_1 \times m_2$, all of which are
``pure'' DMCs. By plugging (\ref{k2_equiv}) into (\ref{k1_equiv}),
it can be seen that
\begin{equation} \label{k1k1lr}
K_1= \sum_{\alpha_1=1, \alpha_2=1}^{\beta_1, \beta_2}
g_{1(\alpha_1)} g_{2(\alpha_2)} R_{1,\alpha_1} R_{2,\alpha_2} K_1
T_{2,\alpha_2} T_{1,\alpha_1}
\end{equation}
is expressed as a convex combination involving the terms
$R_{1,\alpha_1} R_{2,\alpha_2} K_1 T_{2,\alpha_2} T_{1,\alpha_1}$. We first establish the following lemma as an intermediate step.

\begin{lem}
There should be only one term of the form $R_{1,\alpha_1} R_{2,\alpha_2} K_1 T_{2,\alpha_2} T_{1,\alpha_1}$ in the right hand side of (\ref{k1k1lr}), i.e. $\beta_1= \beta_2= 1$, with full-rank $R_{1,\alpha_1} R_{2,\alpha_2}$ and $T_{2,\alpha_2} T_{1,\alpha_1}$.
\end{lem}

\begin{proof}
Let $C_1$ be the capacity of $K_1$, $C_{\alpha_1,\alpha_2}$ be the
capacity of $R_{1,\alpha_1} R_{2,\alpha_2} K_1 T_{2,\alpha_2}
T_{1,\alpha_1}$. Let $I(K,{\bf p})$ denote the mutual information of DMC $K$ with the input distribution represented by row vector ${\bf p}$. Let ${\bf p}^X$ be the capacity-achieving input distribution of $K_1$. We also denote this distribution in terms of the probability mass function (PMF) $p^X(x)$ of $x= 1,\dots,n_1$ as needed. Denote the entry of $R_{1,\alpha_1} R_{2,\alpha_2} K_1 T_{2,\alpha_2} T_{1,\alpha_1}$ with index $(x,y)$ by $p_{\alpha_1,\alpha_2}(y|x)$, considering that they describe transition probabilities. Based on \cite[Theorem 2.7.4]{coverbook06}, we have
\begin{equation} \label{k1k1mi}
C_1 \leq \sum_{\alpha_1=1, \alpha_2=1}^{\beta_1, \beta_2}
g_{1(\alpha_1)} g_{2(\alpha_2)} I_{\alpha_1,\alpha_2}
\end{equation}
Note that $I(R_{1,\alpha_1} R_{2,\alpha_2} K_1 T_{2,\alpha_2} T_{1,\alpha_1}, {\bf p}^X) \leq C_{\alpha_1,\alpha_2}$, and $C_{\alpha_1,\alpha_2} \leq C_1= I(K_1, {\bf p}^X)$. It is clear that if $I(R_{1,\alpha_1} R_{2,\alpha_2} K_1 T_{2,\alpha_2} T_{1,\alpha_1}, {\bf p}^X)< I(K_1, {\bf p}^X)$ for any $\{ \alpha_1,\alpha_2 \}$, it will follow from (\ref{k1k1mi}) that $C_1< C_1$ which is contradictory. Therefore, it is required that $I(R_{1,\alpha_1} R_{2,\alpha_2} K_1 T_{2,\alpha_2} T_{1,\alpha_1}, {\bf p}^X)= I(K_1, {\bf p}^X)$ for all $\{ \alpha_1,\alpha_2 \}$. In what follows, we show that $I(R_{1,\alpha_1} R_{2,\alpha_2} K_1 T_{2,\alpha_2} T_{1,\alpha_1}, {\bf p}^X)< I(K_1, {\bf p}^X)$ holds for the cases in which $R_{1,\alpha_1} R_{2,\alpha_2}$ or $T_{2,\alpha_2} T_{1,\alpha_1}$ is not full-rank, thereby ruling them out.

We first consider what happens if $R_{1,\alpha_1} R_{2,\alpha_2}$ is not full-rank, by comparing $I(R_{1,\alpha_1} R_{2,\alpha_2} K_1, {\bf p}^X)$ with $I(K_1, {\bf p}^X)$. Given the formula \cite[eq. (2.111)]{coverbook06} of mutual information
\begin{equation}
I(X;Y)= H(Y)- \sum_x p(x) H(Y|X=x)
\end{equation}
it is easy to see that $I(R_{1,\alpha_1} R_{2,\alpha_2} K_1, {\bf p}^X)= I(K_1, {\bf p}^X R_{1,\alpha_1} R_{2,\alpha_2})$, since $R_{1,\alpha_1} R_{2,\alpha_2} K_1$ with input distribution ${\bf p}^X$ and $K_1$ with input distribution ${\bf p}^X R_{1,\alpha_1} R_{2,\alpha_2}$ result in the same output ($Y$) distribution, as well as the same row entropy ($H(Y|X=x)$) distribution. With $R_{1,\alpha_1} R_{2,\alpha_2}$ being not full-rank, there should be at least one zero entry in the probability vector ${\bf p}^X R_{1,\alpha_1} R_{2,\alpha_2}$, and ${\bf p}^X R_{1,\alpha_1} R_{2,\alpha_2}$ cannot be a capacity achieving distribution for $K_1$, given assumption {\bf (I)}. On the other hand, based on data processing inequality \cite[Th. 2.8.1]{coverbook06}, we have $I(R_{1,\alpha_1} R_{2,\alpha_2} K_1 T_{2,\alpha_2} T_{1,\alpha_1}, {\bf p}^X) \leq I(R_{1,\alpha_1} R_{2,\alpha_2} K_1, {\bf p}^X)$. Consequently, $I(R_{1,\alpha_1} R_{2,\alpha_2} K_1 T_{2,\alpha_2} T_{1,\alpha_1}, {\bf p}^X) \leq I(R_{1,\alpha_1} R_{2,\alpha_2} K_1, {\bf p}^X)= I(K_1, {\bf p}^X R_{1,\alpha_1} R_{2,\alpha_2})< I(K_1, {\bf p}^X)$, which leads to contradiction as discussed above, and thus $R_{1,\alpha_1} R_{2,\alpha_2}$ must be full-rank for all $\{ \alpha_1,\alpha_2 \}$.

Second, we show that with $R_{1,\alpha_1} R_{2,\alpha_2}$ being full-rank, $I(R_{1,\alpha_1} R_{2,\alpha_2} K_1 T_{2,\alpha_2} T_{1,\alpha_1}, {\bf p}^X)< I(K_1, {\bf p}^X)$ holds if $T_{2,\alpha_2} T_{1,\alpha_1}$ is not full-rank, by comparing $I(R_{1,\alpha_1} R_{2,\alpha_2} K_1 T_{2,\alpha_2} T_{1,\alpha_1}, {\bf p}^X)$ with $I(R_{1,\alpha_1} R_{2,\alpha_2} K_1, {\bf p}^X)$. We use $p(y|x)$ and $p'(y|x)$ to denote the entries of $R_{1,\alpha_1} R_{2,\alpha_2} K_1$ and $R_{1,\alpha_1} R_{2,\alpha_2} K_1 T_{2,\alpha_2} T_{1,\alpha_1}$ with index $(x,y)$ respectively, considering that they describe transition probabilities. It is clear that with $T_{2,\alpha_2} T_{1,\alpha_1}$ being full-rank (and thus a permutation) matrix, $I(R_{1,\alpha_1} R_{2,\alpha_2} K_1 T_{2,\alpha_2} T_{1,\alpha_1}, {\bf p}^X)= I(R_{1,\alpha_1} R_{2,\alpha_2} K_1, {\bf p}^X)$, so we just consider a representative case of $T_{2,\alpha_2} T_{1,\alpha_1}$ being not full-rank: $T_{2,\alpha_2} T_{1,\alpha_1}$ is obtained from switching the $1$ entry at index $(y_1, y_1)$ with the $0$ entry at index $(y_1, y_2)$ in the $m_1 \times m_1$ identity matrix. This results in the relation between $p(y|x)$ and $p'(y|x)$ (for all $x=1,\dots,n_1$) given by: $p'(y_2|x)= p(y_1|x)+ p(y_2|x)$, $p'(y_1|x)=0$, and $p'(y|x)=p(y|x)$ for all other values of $y$ from $1$ through $m_1$. Based on log sum inequality \cite[Th. 2.7.1]{coverbook06}, for all $x=1,\dots,n_1$ we have
\begin{equation} \label{ttfr1}
\begin{split}
& p'(y_2|x) p^X(x) \log \frac{p'(y_2|x)}{\sum_{x=1}^{n_1} p'(y_2|x) p^X(x)}\\
\leq & p(y_2|x) p^X(x) \log \frac{p(y_2|x)}{\sum_{x=1}^{n_1} p(y_2|x) p^X(x)}+ p(y_1|x) p^X(x) \log \frac{p(y_1|x)}{\sum_{x=1}^{n_1} p(y_1|x) p^X(x)}
\end{split}
\end{equation}
and consequently
\begin{equation} \label{ttfr2}
\begin{split}
& \sum_{x=1}^{n_1} p'(y_2|x) p^X(x) \log \frac{p'(y_2|x)}{\sum_{x=1}^{n_1} p'(y_2|x) p^X(x)}\\
\leq & \sum_{x=1}^{n_1} \left[ p(y_2|x) p^X(x) \log \frac{p(y_2|x)}{\sum_{x=1}^{n_1} p(y_2|x) p^X(x)}+ p(y_1|x) p^X(x) \log \frac{p(y_1|x)}{\sum_{x=1}^{n_1} p(y_1|x) p^X(x)} \right]
\end{split}
\end{equation}
Note that the left hand side of (\ref{ttfr2}) makes part of $I(R_{1,\alpha_1} R_{2,\alpha_2} K_1 T_{2,\alpha_2} T_{1,\alpha_1}, {\bf p}^X)$, and the right hand side of (\ref{ttfr2}) makes part of $I(R_{1,\alpha_1} R_{2,\alpha_2} K_1, {\bf p}^X)$, and the remaining terms in the two mutual informations are the same since there is no change made on the output symbols other than $y_1$ and $y_2$, and consequently
\begin{equation} \label{ttfr3}
I(R_{1,\alpha_1} R_{2,\alpha_2} K_1 T_{2,\alpha_2} T_{1,\alpha_1}, {\bf p}^X) \leq I(R_{1,\alpha_1} R_{2,\alpha_2} K_1, {\bf p}^X)
\end{equation}
It is clear that for the equality to hold in (\ref{ttfr3}), the equality needs to hold in (\ref{ttfr1}) for $x=1,\dots,n_1$. Given assumption {\bf (I)} which specifies that $p^X(x)> 0$ for $x=1,\dots,n_1$, it follows that, the equality holds in (\ref{ttfr3}) only when $p(y_2|x)/p(y_1|x)$ is constant for $x=1,\dots,n_1$, or one of $p(y_1|x)$ and $p(y_2|x)$ is zero for $x=1,\dots,n_1$. This leads to the requirement that $K_1$ has a column which is a multiple of another column, or an all-zero column, thereby contradicting assumption {\bf (II)}. Therefore with $T_{2,\alpha_2} T_{1,\alpha_1}$ being not full-rank, strict inequality holds in (\ref{ttfr3}), which in turn leads to $I(R_{1,\alpha_1} R_{2,\alpha_2} K_1 T_{2,\alpha_2} T_{1,\alpha_1}, {\bf p}^X)< I(K_1, {\bf p}^X)$ and $C_1< C_1$ which is contradictory. We have now completed the proof for that $R_{1,\alpha_1} R_{2,\alpha_2}$ and $T_{2,\alpha_2} T_{1,\alpha_1}$ need to be full-rank (and consequently permutation matrices) for all $\{ \alpha_1,\alpha_2 \}$.

Following from the above conclusion, we consider what are further required for equality to hold in (\ref{k1k1mi}), based on log sum inequality \cite[Th. 2.7.1]{coverbook06}. It follows easily from this inequality that, for any $x=1,\dots,n_1$ and $y=1,\dots,m_1$,
\begin{equation} \label{argst1}
\begin{split}
& \left( \sum_{\alpha_1=1, \alpha_2=1}^{\beta_1, \beta_2} g_{1(\alpha_1)} g_{2(\alpha_2)} p_{\alpha_1,\alpha_2}(y|x) p^X(x) \right) \log \frac{\sum_{\alpha_1=1, \alpha_2=1}^{\beta_1, \beta_2} g_{1(\alpha_1)} g_{2(\alpha_2)} p_{\alpha_1,\alpha_2}(y|x)} {\sum_{x=1}^{n_1} \sum_{\alpha_1=1, \alpha_2=1}^{\beta_1, \beta_2} g_{1(\alpha_1)} g_{2(\alpha_2)} p_{\alpha_1,\alpha_2}(y|x) p^X(x)}\\
\leq & \sum_{\alpha_1=1, \alpha_2=1}^{\beta_1, \beta_2} g_{1(\alpha_1)} g_{2(\alpha_2)} p_{\alpha_1,\alpha_2}(y|x) p^X(x) \log \frac{ p_{\alpha_1,\alpha_2}(y|x)} {\sum_{x=1}^{n_1} p_{\alpha_1,\alpha_2}(y|x) p^X(x)}
\end{split}
\end{equation}
It is clear that the summation of (\ref{argst1}) over $x=1,\dots,n_1$ and $y=1,\dots,m_1$ leads to (\ref{k1k1mi}), therefore, for the equality to hold in (\ref{k1k1mi}), it is required that the equality holds in (\ref{argst1}) for all $x=1,\dots,n_1$ and $y=1,\dots,m_1$, which is satisfied only when the $y$-th column of one $R_{1,\alpha_1} R_{2,\alpha_2} K_1 T_{2,\alpha_2} T_{1,\alpha_1}$ term is a multiple of the $y$-th column of another such term, for all $y=1,\dots,m_1$. This in turn requires that ``different'' such terms must be related through diagonal matrices, e.g. it is required that
\begin{equation}
R_{1,1} R_{2,1} K_1 T_{2,1} T_{1,1}= R_{1,1} R_{2,2} K_1 T_{2,2} T_{1,1} D
\end{equation}
with $D$ being a diagonal matrix with the diagonal entries being positive. Considering that $R_{1,1} R_{2,1}$, $T_{2,1} T_{1,1}$, $R_{1,1} R_{2,2}$ and $T_{2,2} T_{1,1}$ are permutation matrices, it follows that $K_1= P_1 K_1 P_2 D$ with $P_1$, $P_2$ being permutation matrices. Given assumption {\bf (III)}, it is required that both $P_1$ and $P_2$ are identity matrices, and also required that $D$ is identity, and $\{ R_{1,\alpha_1} R_{2,\alpha_2}, T_{2,\alpha_2} T_{1,\alpha_1} \}$ are the same for all $\{ \alpha_1,\alpha_2 \}$. Consequently, there should be only one term in the right hand side of (\ref{k1k1lr}). Thus we have proved that $\beta_1= \beta_2= 1$, which in turn implies that we can simplify the notations through $R_{1,\alpha_1}= R_1$, $R_{2,\alpha_2}= R_2$, $T_{1,\alpha_1}= T_1$, $T_{2,\alpha_2}= T_2$.
\end{proof}

Now that we have established ${\rm rank} (R_1 R_2)= n_1$ and ${\rm rank} (T_2 T_1)= m_1$, we consider what implications they have on $R_1$, $R_2$, $T_1$, $T_2$. Given the fact that ${\rm rank} (AB) \leq \min \{ {\rm rank} (A),{\rm rank} (B) \}$, it is further implied that
\begin{equation} \label{rank_rt1}
{\rm rank}(R_1) \geq n_1, {\rm rank}(R_2) \geq n_1, {\rm rank}(T_1)
\geq m_1, {\rm rank}(T_2) \geq m_1
\end{equation}
Similarly, by substituting (\ref{k1_equiv}) into (\ref{k2_equiv}), it
can be derived that
\begin{equation} \label{rank_rt2}
{\rm rank}(R_1) \geq n_2, {\rm rank}(R_2) \geq n_2, {\rm rank}(T_1)
\geq m_2, {\rm rank}(T_2) \geq m_2
\end{equation}
On the other hand, for $R_1$ of size $n_1 \times n_2$,  $R_2$ of
size $n_2 \times n_1$, $T_1$ of size $m_2 \times m_1$, $T_2$ of size
$m_1 \times m_2$, the ranks should satisfy
\begin{equation} \label{rank_rt3}
\begin{split}
& {\rm rank}(R_1) \leq \min (n_1,n_2), {\rm rank}(R_2) \leq \min
(n_1,n_2),\\
& {\rm rank}(T_1) \leq \min (m_1,m_2), {\rm rank}(T_2) \leq \min
(m_1,m_2)
\end{split}
\end{equation}
Given (\ref{rank_rt1}), (\ref{rank_rt2}) and (\ref{rank_rt3}), it
follows that ${\rm rank}(R_1)= {\rm rank}(R_2)= n_1= n_2$ and ${\rm
rank}(T_1)= {\rm rank}(T_2)= m_1= m_2$. Since square full-rank $(0,1)$ matrices are permutation matrices, these four
matrices must be permutation matrices, and in turn it is necessary
to have $K_2= R K_1 T$ with $R$ and $T$ being permutation matrices
for $K_1$ and $K_2$ to be equivalent. It is easy to see that this
condition is also sufficient for the equivalence between $K_1$ and
$K_2$, and the proof is complete.

\section{Proof of Theorem \ref{thr1}} \label{pr_thr1}

We start from (\ref{def_ci}) with $R_{\alpha}$'s and $T_{\alpha}$'s
being pure channels, as equivalent to Definition \ref{defn1}. It is
clear that entries of $K_2$ are linear combinations of the entries
of $K_1$. Considering the fact that there is a one-to-one mapping
between the entries of $K_1$ and the entries of $\mathbf{w}_1$, as
well as the same situation for $K_2$ and $\mathbf{w}_2$, it follows
that there is a matrix $P$ such that $\mathbf{w}_2= P \mathbf{w}_1$,
and the conditions for $K_2 \subseteq K_1$ can be related to what
properties the combining coefficients $[P]_{(i,j)}$'s have. Based on
Birkhoff's Theorem \cite[p.30]{marshallbook09}, both $K_1$ and $K_2$
are inside the convex hull of $n \times n$ permutation matrices,
therefore it is sufficient for $R_{\alpha}$'s and $T_{\alpha}$'s to
contain only permutation matrices; otherwise $\sum_{\alpha=1}
^{\beta} g_{(\alpha)} R_{\alpha} K_1 T_{\alpha}$ will fall out of
the convex hull of $n \times n$ permutation matrices, which
contradicts with the doubly stochastic assumption. Consequently, for
each $\alpha$, $R_{\alpha} K_1 T_{\alpha}$ gives a matrix having
exactly the same set of entries as $K_1$, generated by permuting the
columns and rows of $K_1$. As a result, $R_{\alpha}$'s and
$T_{\alpha}$'s do not replace any row of $K_1$ with the duplicate of
another row, or merge any column into another column and then
replace it with zeros.

We now consider the properties of $[P]_{(i,j)}$'s based on the
structure of $R_{\alpha} K_1 T_{\alpha}$. We have $\sum_j
[P]_{(i,j)}= \sum_{\alpha=1} ^{\beta} g_{(\alpha)}= 1$ for $i=
1,\dots,n^2$ since each entry of $K_1$ is contained in $R_{\alpha}
K_1 T_{\alpha}$ exactly once, for $\alpha= 1,\dots,\beta$. On the
other hand, since each entry $[K_2]_{(i,j)}$ of $K_2$ is the convex
combination of the entries with the same index $(i,j)$ of
$R_{\alpha} K_1 T_{\alpha}$'s, while each entry of $R_{\alpha} K_1
T_{\alpha}$ is exactly an entry of $K_1$, it follows that $\sum_i
[P]_{(i,j)}= \sum_{\alpha=1} ^{\beta} g_{(\alpha)}= 1$ for $j=
1,\dots,n^2$. Also, it is straightforward to see that $[P]_{(i,j)}
\geq 0$ for $i,j=1,\dots,n^2$ due to the non-negativeness of
$g_{(\alpha)}$'s. Consequently, $P$ is doubly stochastic, and
$\mathbf{w}_2= P \mathbf{w}_1$ implies that $\mathbf{w}_2 \prec
\mathbf{w}_1$ \cite[p.155]{marshallbook09}, completing the proof.

\section{Proof of Theorem \ref{thr2}} \label{pr_thr2}

Let $\mathbf{w}_1$ and $\mathbf{w}_2$ be the $n^2 \times 1$ vectors
containing all the entries of $K_1$ and $K_2$ respectively. It is
easy to see that $\mathbf{w}_1$ and $\mathbf{w}_2$ contain the
entries of $\mathbf{v}_1$ and $\mathbf{v}_2$ each duplicated $n$
times respectively. Given $K_2 \subseteq K_1$, based on Theorem
\ref{thr1} we know that $\mathbf{w}_2 \prec \mathbf{w}_1$, thus
$\sum_{i=1}^k n \cdot v_{1(i)}^{\downarrow} \geq \sum_{i=1}^k n
\cdot v_{2(i)}^{\downarrow}$ for $k=1,\dots,n$, and it follows that
$\sum_{i=1}^k v_{1(i)}^{\downarrow} \geq \sum_{i=1}^k
v_{2(i)}^{\downarrow}$ for $k=1,\dots,n$. In addition, $\sum_{i=1}^n
v_{1(i)}= \sum_{i=1}^n v_{2(i)}= 1$ as required for stochastic
matrices, therefore $\mathbf{v}_2 \prec \mathbf{v}_1$ is necessary
for $K_2 \subseteq K_1$.

We next prove that the existence of a probability vector $\mathbf{x}
\in \mathbb{R}^n_+$ such that $\mathbf{v}_1 \circledast \mathbf{x}=
\mathbf{v}_2$ is sufficient for $K_2 \subseteq K_1$. Let $P$ be the
$n \times n$ permutation matrix such that $\mathbf{x} P$ is cyclic
shifted to the right by $1$ with respect to $\mathbf{x}$, and let
$X$ be the $n \times n$ matrix with the $i$-th column being $P^{i-1}
\mathbf{x}^T$. It is easy to see that both $P^{i-1}$ and $X$ are
circulant. Given $\mathbf{v}_1 \circledast \mathbf{x}=
\mathbf{v}_2$, it follows that $\mathbf{v}_1 X= \mathbf{v}_2$ due to
the definition of circular convolution. Also, notice that $P^{i-1}
X= X P^{i-1}$ since the multiplication of two circulant matrices are
commutative. Consequently, the $i$-th row of $K_1$, given by
$\mathbf{v}_1 P^{i-1}$, and the $i$-th row of $K_2$, given by
$\mathbf{v}_2 P^{i-1}$, are related through $(\mathbf{v}_1 P^{i-1})
X= \mathbf{v}_1 X P^{i-1}= (\mathbf{v}_2 P^{i-1})$. It then follows
that $K_1 X= K_2$ with a stochastic matrix $X$, i.e. Definition
\ref{defn1} is satisfied, and the proof is complete.

An alternative proof based on FFT: Let $U$ be the $n \times n$ FFT
matrix. Then $\mathbf{v}_1 \circledast \mathbf{x}= \mathbf{v}_2$
$\Rightarrow$ $FFT(\mathbf{v}_1) \circ FFT(\mathbf{x})=
FFT(\mathbf{v}_2)$ $\Rightarrow$ ${\rm diag}(FFT(\mathbf{v}_1))
 {\rm diag}(FFT(\mathbf{x}))= {\rm diag}(FFT(\mathbf{v}_2))$ $\Rightarrow$ $U {\rm
 diag}(FFT(\mathbf{v}_1)) U^{*} U
 {\rm diag}(FFT(\mathbf{x})) U^{*}= U {\rm diag}(FFT(\mathbf{v}_2))
 U^{*}$ $\Rightarrow$ $K_1 X= K_2$, and the proof is complete.

\section{Proof of Corollary \ref{crl1}} \label{pr_crl1}

It is clear that $3 \times 3$ symmetric DMCs which are not circulant
can have only the following layout:
\begin{equation}
\begin{bmatrix} 1 & 2 & 3 \\ 2 & 3 & 1 \\ 3 & 1 & 2 \end{bmatrix}
\end{equation}
and it can be made circulant by permuting the second and third rows.
Also, $4 \times 4$ symmetric DMCs which are not circulant can have
only the following layouts:
\begin{equation}
\begin{bmatrix} 1 & 2 & 3 & 4 \\ 2 & 1 & 4 & 3 \\ 3 & 4 & 1 & 2 \\ 4 & 3 & 2 & 1
\end{bmatrix}, \begin{bmatrix} 1 & 2 & 3 & 4 \\ 2 & 1 & 4 & 3 \\ 3 & 4 & 2 & 1 \\ 4 & 3 & 1 &
2 \end{bmatrix}, \begin{bmatrix} 1 & 2 & 3 & 4 \\ 3 & 1 & 4 & 2 \\
2 & 4 & 1 & 3 \\ 4 & 3 & 2 & 1 \end{bmatrix},
\end{equation}
together with other layouts obtained by permuting their rows. For
each of these layouts, it is easy to check with MATLAB that there
exists column permutations which can make each of its rows cyclic
shift of the others. Therefore for $n=3,4$, $n \times n$ symmetric
DMCs can be transformed into circulant DMCs. Consequently, the
results in Theorem \ref{thr2} can be applied to circulant DMCs, and
the second statement of the corollary holds.

\section{Proof of Theorem \ref{thr4}} \label{pr_thr4}

Since an $n_2 \times m_2$ stochastic matrix is determined by its
$n_2$ rows and first $m_2-1$ columns, the class of all $n_2 \times
m_2$ stochastic matrices can be viewed as a convex polytope in $n_2
(m_2-1)$ dimensions. We apply Carath\'{e}odory's theorem
\cite[p.155]{rockafellarbook70}, which asserts that if a subset
$\mathbb{S}$ of $\mathbb{R}^m$ is $k$-dimensional, then every vector
in the convex hull of $\mathbb{S}$ can be expressed as a convex
combination of at most $k+1$ vectors in $\mathbb{S}$, on
(\ref{def_ci}) with $K_1$ of size $n_1 \times m_1$ and $K_2$ of size
$n_2 \times m_2$. It is clear that $R_{\alpha} K_1 T_{\alpha}$ and
$K_2$ are at most $n_2 (m_2-1)$-dimensional. Therefore if $K_2$ is
in the convex hull of $\{ R_{\alpha} K_1 T_{\alpha} \}_{\alpha=1}
^{\beta}$, it can be expressed as a convex combination of at most
$n_2 (m_2-1)+1$ matrices in $\{ R_{\alpha} K_1 T_{\alpha}
\}_{\alpha=1} ^{\beta}$, i.e. the number of necessary $\{
R_{\alpha},T_{\alpha} \}$ pairs can be bounded as $\beta= \beta_1
\leq n_2 (m_2-1)+1$ if (\ref{def_ci}) holds. A similar proof can
follow for the case of both $K_1$ and $K_2$ being $n \times n$
doubly stochastic, in which they are at most $(n-1)^2$-dimensional.

\section{Proof of Theorem \ref{thr2_alg2}} \label{pr_thr2_alg2}

First, we prove that at the $t$-th iteration, the $t$-th entry of ${\bf g}_{t}$, denoted by ${\bf g}_{t(t)}$, has the same sign as $\left< {\bf r}_{t-1}, [A_{\rm sel}]_{(:,t)} \right>$ (we use the notation $<\cdot,\cdot>$ for inner product in order to make it clearly identifiable as a scalar), and therefore selecting a negative
inner product would never produce a ${\bf g}_{t} \geq 0$, based on the orthogonality property that ${\bf r}_{t-1}$ is perpendicular to all the columns of $[A_{\rm sel}]_{(:,1:t-1)}$.
Suppose $[A_{\rm sel}]_{(:,t)}$ is the selected column of $A$. Based on
\begin{equation} \label{resid_evo0}
{\bf h}= {\bf r}_{t-1}+ [A_{\rm sel}]_{(:,1:t-1)} {\bf g}_{t-1}=
{\bf r}_{t}+ [A_{\rm sel}]_{(:,1:t)} {\bf g}_{t}
\end{equation}
we have
\begin{equation} \label{resid_evo1}
[A_{\rm sel}]_{(:,1:t)} {\bf g}_{t}= [A_{\rm sel}]_{(:,1:t-1)} {\bf
g}_{t-1}+ {\bf r}_{t-1}- {\bf r}_{t}
\end{equation}
By taking the inner product of (\ref{resid_evo1}) with ${\bf
r}_{t-1}$, we have
\begin{equation} \label{resid_evo2}
\left< {\bf r}_{t-1}, [A_{\rm sel}]_{(:,t)} \right> {\bf
g}_{t(t)}= \left< {\bf r}_{t-1}, {\bf r}_{t-1}-{\bf r}_{t} \right>=
\| {\bf r}_{t-1} \|_2 \left( \| {\bf r}_{t-1} \|_2- \frac{ \left<
{\bf r}_{t-1}, {\bf r}_{t} \right> }{\| {\bf r}_{t-1} \|_2} \right)
\end{equation}
Clearly, $\left< {\bf r}_{t-1}, {\bf r}_{t} \right>/\| {\bf r}_{t-1}
\|_2$ is the scalar projection of ${\bf r}_{t}$ onto ${\bf r}_{t-1}$
and thus $\left< {\bf r}_{t-1}, {\bf r}_{t} \right>/\| {\bf r}_{t-1}
\|_2 \leq \| {\bf r}_{t} \|_2$. In addition, $\| {\bf r}_{t} \|_2<
\| {\bf r}_{t-1} \|_2$ due to the involvement of an additional
column in the least-square problem. Consequently, (\ref{resid_evo2})
is positive, which implies that ${\bf g}_{t(t)}$ has the same sign as $\left< {\bf r}_{t-1}, [A_{\rm sel}]_{(:,t)} \right>$ and $\left< {\bf r}_{t-1}, [A_{\rm sel}]_{(:,t)}
\right> >0$ is necessary for ${\bf g}_{t} \geq {\bf 0}$.

Second, we prove that it is always possible to select a column
$[A_{\rm sel}]_{(:,t)}$ from $A$, such that $\left< {\bf r}_{t-1},
[A_{\rm sel}]_{(:,t)} \right> >0$, before the iterations terminate
(i.e. ${\bf r}_{t} \ne 0_{p \times 1}$). Define three sets of $p
\times 1$ vectors $S_1= \{ {\bf v}|\left< {\bf r}_{t}, {\bf v}
\right> >0 \}$, $S_2= \{ {\bf v}|\left< {\bf r}_{t}, {\bf v} \right>
<0 \}$ and $S_3= \{ {\bf v}|\left< {\bf r}_{t}, {\bf v} \right> =0
\}$. It is clear that $S_1$, $S_2$ and $S_3$ are mutually exclusive
and are all convex. It is also clear that all the columns of
$[A_{\rm sel}]_{(:,1:t)}$ is in $S_3$, and ${\bf h} \in S_1$ based
on (\ref{resid_evo0}). If there is no column of $A$ which is in
$S_1$, then ${\bf h}$ cannot be in the convex hull of the columns of
$A$, hence contradicting the fact that $A {\bf g}= {\bf h}$ with
some probability vector ${\bf g}$. Therefore a positive inner
product together with its corresponding column of $A$ is always
available for selection, and the proof is complete.

\section{Proof of Theorem \ref{thr1_alg2}} \label{pr_thr1_alg2}

We first start with two preparatory lemmas which generalize Conjecture \ref{prop1}.

\begin{lem} \label{prop1_ext1}
Let $n \times m$ matrix $G$ have all of its entries being non-negative and all of its columns being linearly independent, and $G {\bf x}_1= {\bf x}_2$
with all entries of vector ${\bf x}_1$ being non-negative, then
there exists at least one column ${\bf g}_*$ of $G$ such that, with
$G_*$ obtained by excluding ${\bf g}_*$ from $G$,
\begin{equation} \label{x3_pos}
{\bf x}_3= \arg\min_{\bf x} \| {\bf x}_2- G_* {\bf x} \|_2^2= \arg\min_{\bf x} \| G {\bf x}_1- G_* {\bf x} \|_2^2 \geq {\bf 0}
\end{equation}
\end{lem}

\begin{proof}
According to Conjecture \ref{prop1}, there exists at least one
column ${\bf g}_*$ of $G$ such that, with $G_*$ obtained by
excluding ${\bf g}_*$ from $G$,
\begin{equation} \label{x4_pos}
G_*^T (g_*- G_* {\bf x}_4)= 0
\end{equation}
holds with ${\bf x}_4 \geq 0$. Let ${\bf y}_1$ be the part of ${\bf
x}_1$ corresponding to $G_*$ and $y_g$ be the part of ${\bf x}_1$
corresponding to ${\bf g}_*$, i.e.
\begin{equation} \label{x2_decomp}
{\bf x}_2= G {\bf x}_1= G_* {\bf y}_1+ y_g {\bf g}_*
\end{equation}
The assertion in (\ref{x3_pos}) is equivalent to the
existence of ${\bf x}_3 \geq 0$ such that $G_*^T G_* {\bf x}_3=
G_*^T {\bf x}_2= G_*^T G {\bf x}_1$. In order to prove this, according to Farkas' lemma \cite[Proposition 1.8]{gunterbook95} which states a sufficient condition for such non-negative vector to exist, it suffices to prove that for
any vector ${\bf x}_5$ such that $(G_*^T G_*)^T {\bf x}_5= G_*^T G_*
{\bf x}_5 \geq 0$, $(G_*^T G {\bf x}_1)^T {\bf x}_5= {\bf x}_1^T G^T
G_* {\bf x}_5 \geq 0$ holds. Based on (\ref{x4_pos}) and (\ref{x2_decomp}), it is clear that
\begin{equation}
{\bf x}_1^T
G^T G_* {\bf x}_5= {\bf x}_2^T G_* {\bf x}_5= (G_* {\bf y}_1+ y_g {\bf g}_*)^T G_* {\bf x}_5= {\bf y}_1^T G_*^T G_* {\bf x}_5+ y_g {\bf x}_4^T G_*^T G_* {\bf x}_5 \geq 0
\end{equation}
given the known conditions that $G_*^T G_* {\bf x}_5 \geq 0$, ${\bf x}_4 \geq 0$, ${\bf y}_1 \geq 0$ and $y_g \geq 0$, and we have proved (\ref{x3_pos}) which generalizes Conjecture \ref{prop1}.
\end{proof}

\begin{lem} \label{prop1_ext2}
There exists a set of matrices $\{ G_{k} \}_{k=1}^{m-1}$, in which $G_{m-1}$ is obtained by excluding one column from $G$ and $G_{k}$ is obtained by excluding one column from $G_{k+1}$ for $k=1,\dots,m-2$, such that (\ref{x3_pos}) holds with $G_*$ replaced by any matrix in $\{ G_{k} \}_{k=1}^{m-1}$, i.e. $\arg\min_{\bf x} \| {\bf x}_2- G_{k} {\bf x} \|_2^2 \geq {\bf 0}$ with ${\bf x}_2= G {\bf x}_1$.
\end{lem}

\begin{proof}
Equation (\ref{x3_pos}) implies that, for ${\bf x}_2= G {\bf
x}_1$, there exists at least one column ${\bf g}_*$ of $G$ such
that, with $G_{m-1}$ obtained by excluding ${\bf g}_*$ from $G$, the
orthogonal projection $G_{m-1} {\bf x}_3$ of ${\bf x}_2$ onto the column
space of $G_{m-1}$ is inside the convex cone generated by the columns of
$G_{m-1}$. Furthermore, for any matrix $G_{m-2}$ whose columns form a
subset of the columns of $G_{m-1}$, it is easy to notice that the
orthogonal projection of $G_{m-1} {\bf x}_3$ onto the column space of
$G_{m-2}$ is identical to the orthogonal projection of ${\bf x}_2$
onto the column space of $G_{m-2}$, i.e. $\arg\min_{\bf x} \| {\bf x}_2- G_{m-2} {\bf x} \|_2^2= \arg\min_{\bf x} \| G_{m-1} {\bf x}_3- G_{m-2} {\bf x} \|_2^2$. Considering that ${\bf x}_3 \geq {\bf 0}$ as proved for (\ref{x3_pos}) above, it follows that there exists a matrix $G_{m-2}$ obtained by excluding a column from $G_{m-1}$, such that $\arg\min_{\bf x} \| G_{m-1} {\bf x}_3- G_{m-2} {\bf x} \|_2^2 \geq {\bf 0}$, and thus $\arg\min_{\bf x} \| {\bf x}_2- G_{m-2} {\bf x} \|_2^2 \geq {\bf 0}$. This in turn implies that by excluding the columns of $G$ one by one, it is always possible to guarantee that
the orthogonal projection of ${\bf x}_2$ onto the linear space
formed by the remaining columns, is inside the convex cone generated
by the remaining columns, at each step, i.e. there exists $\{ G_{k} \}_{k=1}^{m-1}$, in which $G_{m-1}$ is obtained by excluding one column from $G$ and $G_{k}$ is obtained by excluding one column from $G_{k+1}$ for $k=1,\dots,m-2$, such that $\arg\min_{\bf x} \| {\bf x}_2- G_k {\bf x} \|_2^2 \geq {\bf 0}$ for $k=1,\dots,m-1$.
\end{proof}

%%%%%%%%%%%%%%%%%%%%%%%%%%%%%%%%%%%%%%%%%%%%%%%%%%%%%%%%%%%%%%%%%

Note that since $\arg\min_{\bf x} \| {\bf x}_2- G_{k} {\bf x} \|_2^2 \geq {\bf 0}$ is not affected by permuting the columns of $G_{k}$, Lemma \ref{prop1_ext2} can be alternatively stated as follows: there exists at least one column permuted version of $G$, say $G_{m}$, such that $\arg\min_{\bf x} \| G {\bf x}_1- [G_{m}]_{(:,1:k)} {\bf x} \|_2^2 \geq {\bf 0}$ holds for $k=1,\dots,m-1$. This will be applied to prove that Algorithm \ref{alg2} can successfully find a sparse probability vector involved in channel inclusion. Here we use the notations in the description of Algorithm \ref{alg2}, and also new notations as needed. With the presence of inclusion and the actual sparsity level being $s_1$, there exists at least one $p \times s_1$ matrix $A_{\rm sel}^*$ with linearly independent columns, together with $s_1 \times 1$ vector ${\bf g}_{s_1}^* \geq {\bf 0}$, such that $A_{\rm sel}^* {\bf g}_{s_1}^*= {\bf h}$, where $A_{\rm sel}^*$ is defined as a $p \times s_1$ matrix whose columns form a subset of the columns of $A$. Accordingly, there exists at least one column permuted version of $A_{\rm sel}^*$, say $A_{s_1}$, such that
\begin{equation} \label{asel_col_perm}
\arg\min_{\bf g} \| {\bf h}- [A_{s_1}]_{(:,1:k)} {\bf g} \|_2^2 \geq {\bf 0}
\end{equation}
holds for $k=1,\dots,s_1-1$, and also for $k=s_1$ since $A_{s_1} {\bf g}'_{s_1}= {\bf h}$ with ${\bf g}'_{s_1}$ being some entry-permuted version of ${\bf g}_{s_1}^*$. This fact will be used in the following to make categorization of the possible behaviors of Algorithm \ref{alg2} in terms of attempts made on the columns of $A$, from beginning ($t_{\rm act}=0$, $t=1$) to termination (when either a sparse solution is found giving $f=1$, $t= s_1+1$, or the algorithm declares no solution being found giving $f=0$, $t=0$), which will lead to the conclusion that all possible behaviors of Algorithm \ref{alg2} lead to $f=1$.

Mathematically, the behavior of Algorithm \ref{alg2} in terms of attempts made on the columns of $A$ from beginning to termination is defined as this: it is an ordered set $\mathcal{B}$ which has $s_1$ elements, and the $k$-th element $\mathcal{B}_k$ itself is a set with the elements being the columns of $A$ that was attempted for the selection of $[A_{\rm sel}]_{(:,k)}$, for $k=1,\dots,s_1$. Specifically, if some column $[A]_{(:,j)}$ of $A$ was attempted for the selection of $[A_{\rm sel}]_{(:,k)}$, then $[A]_{(:,j)} \in \mathcal{B}_k$, otherwise $[A]_{(:,j)} \not\in \mathcal{B}_k$. Before making the proposed categorization, we establish some useful preliminaries. We refer to $f=1$ as success and $f=0$ as failure when Algorithm \ref{alg2} terminates. Note that $t$ is a function of $t_{\rm act}$ and will be denoted by $t(t_{\rm act})$ as needed for clarification. The term ``residue'' will refer to the residue associated with the $t$ columns $[A_{\rm sel}]_{(:,1:t)}$ (which are already selected), i.e. ${\bf h}- [A_{\rm sel}]_{(:,1:t)} \arg\min_{\bf g} \| {\bf h}- [A_{\rm sel}]_{(:,1:t)} {\bf g} \|_2^2$. We define the notion of order-$t$ generalized failure with specified $[A_{\rm sel}]_{(:,1:t)}$ (i.e. already selected $t$ columns satisfying $\arg\min_{\bf g} \| {\bf h}- [A_{\rm sel}]_{(:,1:t)} {\bf g} \|_2^2 \geq {\bf 0}$), as the situation in which all the (remaining) columns of $A$ having positive inner product with the residue are attempted but not selected (or selected and removed later) as $[A_{\rm sel}]_{(:,t+1)}$ and backtracking has to be performed, i.e. $t(t_{\rm act}+1)= t(t_{\rm act})- 1$, as reflected by Lines $3$ and $16$ in Algorithm \ref{alg2}. Here we allow $t=0$ and treat $[A_{\rm sel}]_{(:,1:0)}$ as an empty ($p \times 0$) matrix accordingly, hence order-$0$ generalized failure is equivalent to failure. Clearly, a generalized failure does not necessarily lead to a failure, unless it is order-$0$, making it a necessary but not sufficient condition of failure. Therefore, with specified $[A_{\rm sel}]_{(:,1:t)}$, by ruling out the possibility of order-$t$ generalized failure, it can be established that Algorithm \ref{alg2} should result in success with such $t$ columns specified. Also, notice that for some order-$t$ generalized failure to occur, a necessary condition is that all possible choices for $[A_{\rm sel}]_{(:,t+1)}$ (i.e. all remaining columns of $A$ having positive inner product with the residue) are attempted, and consequently this becomes a necessary condition for failure to occur. Furthermore, due to the backtracking feature, it is possible for Algorithm \ref{alg2} to attempt any possible choice for $[A_{\rm sel}]_{(:,t+1)}$ (i.e. the columns of $A$ having positive inner product with the residue). Based on Theorem \ref{thr2_alg2}, if some choice of $[A_{\rm sel}]_{(:,t+1)}$ appended to $[A_{\rm sel}]_{(:,1:t)}$ results in non-negative LS solution, then such choice should have positive inner product with the residue. These further imply that with specified $[A_{\rm sel}]_{(:,1:t)}$, if some column $[A]_{(:,j)}$ of $A$ appended to $[A_{\rm sel}]_{(:,1:t)}$ results in non-negative LS solution, i.e. with $[A_{\rm sel}]_{(:,t+1)}= [A]_{(:,j)}$, $\arg\min_{\bf g} \| {\bf h}- [A_{\rm sel}]_{(:,1:t+1)} {\bf g} \|_2^2 \geq {\bf 0}$, but $[A]_{(:,j)}$ has never been attempted for the selection of $[A_{\rm sel}]_{(:,t+1)}$ after the termination of Algorithm \ref{alg2}, then Algorithm \ref{alg2} should result in success with such specified $[A_{\rm sel}]_{(:,1:t)}$. With this implication utilized below, the possible behaviors of Algorithm \ref{alg2} will be categorized into the ones that lead to success for sure and the ones that may lead to generalized failures, in a recursive manner, and we finally rule out the possibility of generalized failures.

Let $P_0$ denote the set of all possible $\mathcal{B}$'s, i.e. all possible behaviors in terms of column attempts of Algorithm \ref{alg2}. Define $P_k$ as the set of possible behaviors of Algorithm \ref{alg2} with $[A_{s_1}]_{(:,1:k)}$ specified as $[A_{\rm sel}]_{(:,1:k)}$, for $k=0,\dots,s_1$, which is in accordance with what $P_0$ represents and induction will be enabled. Let $N_{k+1}$ denote the subset of $P_k$ in which $[A_{s_1}]_{(:,k+1)} \not\in \mathcal{B}_{k+1}$, i.e. $[A_{s_1}]_{(:,k+1)}$ was never attempted for the selection of $[A_{\rm sel}]_{(:,k+1)}$, for $k=0,\dots,s_1-1$. For the base case, we consider the attempts made on selecting the first column of $A_{\rm sel}$, which happen at the instants with $t=1$, regardless of what $t_{\rm act}$ is, as reflected by $\mathcal{B}_1$. We categorize $P_0$ into $P_0= P_1 \cup N_1$ with $P_1 \cap N_1= \emptyset$, according to whether $[A_{s_1}]_{(:,1)} \in \mathcal{B}_1$ or not: $P_1$ denotes the subset of $P_0$ in which $[A_{s_1}]_{(:,1)} \in \mathcal{B}_1$, i.e. $[A_{s_1}]_{(:,1)}$ was attempted at $t=1$, $N_1$ denotes the subset of $P_0$ in which $[A_{s_1}]_{(:,1)} \not\in \mathcal{B}_1$, i.e. $[A_{s_1}]_{(:,1)}$ was never attempted at $t=1$. Based on the above mentioned implication, since (\ref{asel_col_perm}) is satisfied with $k=1$, $N_1$ gives rise to success, and $P_1$ gives rise to the selection of $[A_{s_1}]_{(:,1)}$ as $[A_{\rm sel}]_{(:,1)}$ (this can be justified based on Lines $9$ and $12$ in Algorithm \ref{alg2}). At this stage, we have some doubt if $P_1$ will lead to some generalized failure, while such possibility will eventually be ruled out as we perform further categorization on $P_1$. For the inductive step, consider the attempts made on selecting the $(k+1)$-th column of $A_{\rm sel}$ (with $[A_{\rm sel}]_{(:,1:k)}$ already specified), which happen at the instants with $t=k+1$, regardless of what $t_{\rm act}$ is. It can be easily verified that, the complimentary set of $N_{k+1}$ in $P_k$ is $P_{k+1}$, since based on (\ref{asel_col_perm}), $[A_{s_1}]_{(:,k+1)}$ will be selected as $[A_{\rm sel}]_{(:,k+1)}$ if it is attempted. Thus we now have $P_k= P_{k+1} \cup N_{k+1}$ with $P_{k+1} \cap N_{k+1}= \emptyset$ for $k=1,\dots,s_1-1$, and eventually we have
\begin{equation} \label{cat1_alg2}
P_0= \cup_{k=1}^{s_1} N_k \cup P_{s_1}
\end{equation}
with the individual sets on the right hand side being mutually exclusive. Similar to the case of $N_1$, each $N_k$ in (\ref{cat1_alg2}) gives rise to success. It is also clear that $P_{s_1}$ gives rise to success, since it has $A_{s_1}$ specified as $A_{\rm sel}$, and $A_{s_1} {\bf g}'_{s_1}= {\bf h}$ holds with ${\bf g}'_{s_1} \geq {\bf 0}$. It then follows that $P_0$ gives rise to success, i.e. Algorithm \ref{alg2} is able to find a sparse probability vector successfully when inclusion is present.

%\balance
\bibliographystyle{IEEEtran}
\bibliography{reflist}

\end{document}